\documentclass[journal]{ieeetran}
\usepackage{cite}
\usepackage{amsmath,amssymb,amsfonts}
\usepackage{amsthm}
\usepackage{algorithm}
\usepackage{adjustbox}
\usepackage{algorithm}
\usepackage{algpseudocode}

\usepackage[font=small]{caption}

\usepackage{graphics} 
\usepackage{graphicx}
\usepackage{epsfig} 

\usepackage{tikz}
 \usetikzlibrary{automata,chains}
 
\usepackage{float}

\newtheorem{thm}{Theorem}
\newtheorem{lem}[thm]{Lemma}
\newtheorem{prop}[thm]{Proposition}
\newtheorem{coro}[thm]{Corollary}

\newtheorem{prob}{Problem}

\begin{document}
\title{Stealthy hacking and secrecy of controlled state estimation systems with random dropouts}
\author{Jingyi Lu, Daniel E. Quevedo, Vijay Gupta, and Subhrakanti Dey
\thanks{J. Lu is with the Department of Electrical Engineering (EIM-E), Paderborn University, Paderborn, Germany. E-mail:
	jingyi.lu@upb.de. }
\thanks{D. E. Quevedo is with the School of Electrical Engineering and Robotics, Queensland University of Technology (QUT), Brisbane, Australia. Email: daniel.quevedo@qut.edu.au.}
\thanks{V. Gupta is with Department of Electrical Engineering, University of Notre Dame, Notre Dame, USA. E-mail: vgupta2@nd.edu.}
\thanks{S. Dey is with Hamilton Institute, National University of Ireland, Maynooth, Ireland.
	E-mail: Subhra.Dey@signal.uu.se.}
}

\maketitle

\begin{abstract}
We study the maximum information gain that an adversary may obtain through hacking without being detected. Consider a dynamical process observed by a sensor that transmits a local estimate of the system state to a remote estimator according to some reference transmission policy across a packet-dropping wireless channel equipped with acknowledgments (ACK). An adversary overhears the transmissions and proactively hijacks the sensor to reprogram its transmission policy. We define perfect secrecy as keeping the averaged expected error covariance bounded at the legitimate estimator and unbounded at the adversary. By analyzing the stationary distribution of the expected error covariance, we show that perfect secrecy can be attained for unstable systems only if the ACK channel has no packet dropouts. In other situations, we prove that independent of the reference policy and the detection methods, perfect secrecy is not attainable. For this scenario, we formulate a constrained Markov decision process to derive the optimal transmission policy that the adversary should implement at the sensor, and devise a Stackelberg game to derive the optimal reference policy for the legitimate estimator. 
\end{abstract}


\section{Introduction}
\label{sec:introduction}

\IEEEPARstart{C}{yber}-physical systems (CPSs) have been widely applied to many critical infrastructures including smart grids, transportations, and industrial control systems et al. CPSs can provide efficiency and versatility by merging computing and communication with the physical world. However, the vulnerability of CPSs to various cyber-attacks, which compromise measurements and actuator data availability, integrity, and confidentiality, brings huge potential security and privacy issues\cite{chong2019tutorial}.  Numerous malware targeting industrial control systems have been discovered, such as the notorious Stuxnet, Triton, and Havex \cite{ma2019stealthy,lee2020keeping}. By launching a phishing attack to gain remote access to critical components in an autonomous system, the malware is capable of remote monitoring and potentially taking full control of the target for espionage and sabotage \cite{kushner2013real}. This will often circumvent a digital safety system, thus, potentially serving as a detrimental cyber weapon of mass destruction. 


Considering system security and privacy of dynamical systems, a surge of research has been carried out to investigate attack patterns. For example, optimal jamming attacks targeting data availability were discussed in 
\cite{zhang2015optimal} and \cite{zhang2015optimal2} within the framework of Linear Quadratic Gaussian control and remote state estimation. False data injection attacks in electric energy systems  were studied in \cite{xie2010false} and \cite{liu2011false} to destroy integrity. Eavesdropping attacks were explored in \cite{leong2018transmission} and \cite{tsiamis2019state} for confidentiality violation. From the perspective of a defender, strategies based on encryption \cite{darup2018towards,alexandru2019encrypted} and hypothesis testing \cite{mo2009secure,ren2018binary} were proposed for prevention and detection of the attacks. Resilient approaches that take into account the interactive decision-making process between an agent and its adversary were developed through game theory to counter the malicious attacks \cite{miao2013stochastic,yuan2013resilient,li2015jamming,yuan2016resilient}.

In this work, we focus on confidentiality issues in remote state estimation of dynamical systems. In our setup, a sensor takes measurements of the process and forwards its local estimate of the system state to a remote estimator across a packet-dropping channel; an eavesdropper overhears the transmission to intercept system information (See Fig. \ref{setup1}). This setup was investigated in \cite{ leong2018use,leong2019information,tsiamis2017state,ding2020remote, yang2020encoding}. To preserve privacy for the legitimate estimator, Tsiamis et al. \cite{tsiamis2017state} introduced a control-theoretic definition of ``perfect secrecy'' which requires that the user's expected error remains bounded while the eavesdropper’s expected error grows unbounded. The authors proved that by applying a secrecy mechanism that randomly withholds sensor information, perfect secrecy can be ensured on the premise that the user's packet reception rate is
larger than the adversary's. Our previous work in  \cite{leong2018transmission} shows that the perfect secrecy can be achieved without the prerequisite on the packet reception rate, provided acknowledgments on successful reception at the legitimate estimator are available to the sensor. The work \cite{leong2018transmission} devised a deterministic threshold-type event-triggered policy to realize perfect secrecy. Beyond that, schemes that use artificial noise, encoding-decoding mechanism, and a combination of the two were investigated in \cite{leong2018use, tsiamis2019state,yang2020encoding} for protection of the state information. These works focus on the prevention of information leakage from the position of legitimate estimator. From the adversary's point of view, our recent work \cite{ding2020remote} proposed an active eavesdropping scheme assuming that the eavesdropper can overhear both the transmission and acknowledgment channel, and is capable of switching between a jamming mode and an eavesdropping mode. By carefully scheduling the two modes, a refined estimation performance can be achieved at the eavesdropper.

Motivated by the mechanism of the malware Stunex and Havex, in the present work, we consider a more capable adversary that can eavesdrop on the transmission channel and hack the sensor to reprogram the transmission policy (see Fig. \ref{setup}). This type of adversary may be interested in degrading the performance of the remote estimator, as investigated in \cite{cheng2019event}. In our current work, we focus on an adversary that intends to overhear information in order to form its own state estimate, rather than bringing loss to the legitimate estimator. Note that in contrast to the active eavesdropper in \cite{ding2020remote}, the adversary considered in the current formulation, whilst having a similar goal,   is more capable since the jamming strategies of \cite{ding2020remote}  can also be realized as a special case by withholding certain transmissions. In addition, note that without considering stealthy issues, the adversary in this work can achieve the minimal averaged expected error covariance (AEEC) after launching an intrusion to the sensor which is impossible to the active eavesdropper in \cite{ding2020remote}. In this sense, the attack scenario considered in the present work is more advantageous to the adversary and more challenging to the legitimate party. 

\begin{figure}[t]
	\centering
	\includegraphics[width=0.4\textwidth]{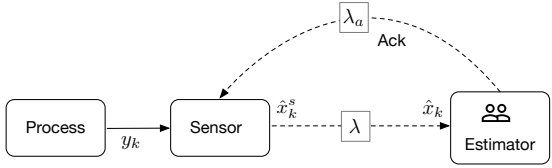}
	\caption{Remote state estimation across packet-dropping channels.} \label{setup1}
\end{figure}


In this paper, we study the maximum information gain that an adversary may obtain through hacking without being detected.  Our aim is to elucidate the tradeoffs between information leakage and detection efforts. Specifically, to derive the optimal malicious policy from the adversary's point of view, we model the transmission process as a constrained Markov decision process (MDP) where the adversary's AEEC is minimized and a stealthy constraint is incorporated. Different from \cite{ding2020remote} and \cite{cheng2019event}, where stealthy constraints are designed for some specific detection algorithms,  we adopt a notion of stealthiness that is independent of any particular detection algorithm being employed. More precisely, we focus on an information-theoretic quantity relating the marginal distribution of the observations at the legitimate estimator before and after the intrusion according to the MDP discussed above. Similar to \cite{bai2015security,bai2017data,che2013reliable,leong2019game}, this notion allows us to obtain a fundamental bound on the stealthiness of the adversary by considering the underlying hypothesis detection problem of identifying if an adversary is present, rather than limiting ourselves to specific detection algorithms used to solve that detection problem. In this framework, we show by analyzing the convergence properties of the stationary distribution that the reliability of the acknowledgment (ACK) channel largely determines the confidentiality. For situations where the ACK channel has no packet dropouts, we prove the existence of a reference policy making the adversary's AEEC unbounded even after the intrusion. In contrast, if the ACK channel has packet dropouts, we show that there always exists a malicious policy that ensures bounded AEEC at the adversary against any stealthy tolerance and any reference policy that gives a bounded AEEC at the legitimate estimator. This property also guarantees the feasibility of the constrained MDP considered and allows us to solve the infinite-state MDP with a finite linear program. 

\begin{figure}[t]
	\centering	\includegraphics[width=0.4\textwidth]{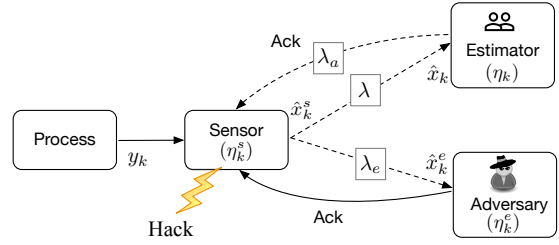}
	\caption{Problem setup: A sensor transmits the local estimate to a legitimate remote estimator. An adversary overhears the transmissions for espionage. It is  capable of (1) hacking the sensor; (2) sending an acknowledgement to the sensor; (3) reprograming the transmission policy according to the acknowledgments from the estimator and the adversary.} \label{setup}
\end{figure}

As a defense mechanism, we explore a resilient design of the optimal reference policy from the perspective of the legitimate estimator. The estimator seeks to optimize its estimation performance and reduce information leakage. 
Assuming that the adversary always commits to the optimal malicious policy for a given reference policy, we formulate this hierarchical decision-making problem as a static Stackelberg game, where the legitimate estimator acts as a leader and the adversary acts as a follower. This results in a bilevel program. In particular, its lower level takes the parameters of the malicious policy as the decision variables and minimizes the adversary's AEEC subject to the stealthy constraint. Its upper level takes the parameters of the reference policy as the decision variable and minimizes a linear combination of the estimator's AEEC and the negative of the adversary's AEEC. Meanwhile, the lower level problem is embedded as a constraint into the upper-level problem. Bilevel programmings are widely known to be hard due to their intrinsic non-convexity and non-differentiability \cite{sinha2017review}. 
In this paper, we can simplify the problem aided by the duality theorem \cite{anandalingam1990solution}, the linear structure of the objective function, and the special shape of the feasible region. We reformulate the bilevel problem as a biconvex program \cite{gorski2007biconvex}, whose optima stay at extreme points. Based on this property, the optimal solution can be achieved by enumerating the extreme points. 
A depth-first branch-and-bound algorithm is further devised to accelerate the enumeration. Numerical examples are provided for validation.

The main contributions of this paper can be briefly summarized as follows.
\begin{itemize}
	\item We study a new attack policy for system confidentiality in remote state estimation. An  optimal covariance-dependent malicious transmission policy is presented in the framework of a constrained Markov decision process. 
	\item We analyze the decisive role of the reliability of the ACK channel in system confidentiality. We show that if the ACK channel has no packet dropouts, then perfect secrecy can be achieved under a stealthy attack; otherwise, perfect secrecy can never be realized.  
	\item From the legitimate estimator's perspective, we derive an optimal resilient reference policy by formulating a suitable  bilevel program. This is then solved by our proposed depth-first branch-and-bound algorithm.
\end{itemize}

The remaining parts of the paper are arranged as follows: Section \ref{sec2} provides the system model and problem setup; Section \ref{sec3} analyzes the performance at the adversary before an intrusion is launched; Section \ref{sec4} details the design of the stealthy constraint and derivation of the optimal intrusion policy; Section \ref{sec5} formulates a bilevel program for the synthesis of the optimal reference policy; Section \ref{sec6} presents the results of the numerical examples; finally, Section \ref{sec7} concludes.

Notations: $\mathbb{R}^n$ is the set of $n$ dimensional vectors. $\mathbb{N}$ is the set of natural numbers. $X\in\mathbb{R}^{n_1\times n_2}$ indicates that $X$ is a $n_1$ by $n_2$ matrix. Given that $X\in\mathbb{R}^{n\times n}$, $X\geq 0$ denotes that the matrix $X$ is positive semidefinite. $\text{Tr}(X)$ refers to the trace of $X$. $\sigma_{\max}(X)$, $\lambda_{\max}(X)$, $\lambda_{\min}(X)$ denote the singular value, the largest eigenvalue, and the smallest eigenvalue, respectively.  $\mathbb{P}(Y\mid z)$ denotes the probability of $Y$ conditional on $z$. $\mathbb{E}(Y)$ is the mathematical expectation.

\section{PROBLEM FORMULATION}
\label{sec2}
\subsection{System model}

Consider a discrete-time linear system as
\begin{align}
x_{k+1} = A x_{k} + w_k, \label{sys}
\end{align}
where $x_k\in\mathbb{R}^{n_s}$ is the system state. The state $x_k$ is measured by a sensor as
\begin{align}
y_k = C x_k + v_k,\label{sys2}
\end{align}
with $y_k\in\mathbb{R}^{n_y}$ denoting the measurement. $w_k\in\mathbb{R}^{n_s}$ and $v_k\in\mathbb{R}^{n_y}$ are process and measurement noises. Assume $w_k$ and $v_k$ are mutually independent Gaussian process with zero mean. The covariance of $w_k$ and $v_k$ are denoted as $Q\in\mathbb{R}^{n_s\times n_s}$ and $R\in\mathbb{R}^{n_y\times n_y}$. Moreover, it is assumed that the pair $(C,A)$ is observable and $(A,\sqrt{Q})$ is controllable.

\par As depicted in Fig. \ref{setup1}, we consider a sensor makes measurements of the system output and optimally estimates system states with a Kalman filter. Denote the posterior local state estimate as 
\begin{align}
\hat{x}_{k}^s \triangleq \mathbb{E}[x_k\mid y_0,\dots,y_k].\nonumber
\end{align}
The corresponding error covariance $P_k^s\triangleq \mathbb{E}[(x_k-x_k^s)(x_k-x_k^s)^\top\mid y_0,\dots,y_k]$ exponentially converges to a steady state \cite{kailath2000linear}. Denote the steady state value as $\overline{P}$. After obtaining $\hat{x}_k^s$, the sensor broadcasts the estimate to a remote estimator. Denote an indicator variable $\nu_k$ as the transmission command such that the sensor transmits if $\nu_k=1$ and keeps silent if $\nu_k=0$. Let $\gamma_k$ be an indicator variable for successful reception, i.e. $\gamma_k=1$
denoting a successful reception and $\gamma_k=0$ implying the occurrence of a packet dropout. Assume $\gamma_k$ is i.i.d Bernoulli, namely,
\begin{align}
\mathbb{P}[\gamma_k=1\mid \nu_k=1]=\lambda\label{trans1}.
\end{align}
Define $\mathcal{I}_k$ as a collection of historical information, i.e.
\begin{align}\mathcal{I}_k\triangleq\{\nu_0\gamma_0,\dots,\nu_k\gamma_k, \nu_0\gamma_0\hat{x}_0^s,\dots,\nu_k\gamma_k\hat{x}_l^s\}.\nonumber
\end{align}
The optimal estimate at the legitimate estimator, denoted as $\hat{x}_k$, is in the form of
\begin{align}
\hat{x}_{k}   \triangleq \mathbb{E}[x_k\mid \mathcal{I}_k]= \left\{\begin{array}{cl} \hat{x}_{ k|k}^s, & 
\gamma_{k} \nu_k= 1 \\ A \hat{x}_{k-1}, & 	\gamma_{k} \nu_k= 0.
\end{array}  \right. \label{estimate}
\end{align} 
Correspondingly, the error covariance at the estimator, denoted as $P_k$, is in the form of 
\begin{align}
P_{k} \triangleq& \mathbb{E}[(x_k-\hat{x}_k)(x_k-\hat{x}_k)^\top\mid \mathcal{I}_k]\nonumber\\
= &\left\{\begin{array}{cl}  \overline{P}  &   \textnormal{
	$\gamma_{k} \nu_k=  1$} \\ 
f(P_{k-1}),  &
\textnormal{
	$\gamma_{k} \nu_k=  0$} \end{array} \right. . \label{error}
\end{align}
Here 
\begin{align}
f(X) \triangleq A X A^T + Q. \label{def_f}
\end{align}
Denote $\eta_k$ as the holding time at the receiver, which is defined as the number of time steps since the last successful receipt, i.e.
\begin{align}
\eta_k  \triangleq \min \{\eta\geq 0: \gamma_{k-\eta}=1\}.\label{hold}
\end{align} 
$\eta_k$ satisfies the recursion
\begin{align}
	\eta_k=\left\{\begin{array}{cc}
		0 & \nu_k\gamma_k=1\\
		\eta_{k-1}+1 & \nu_k\gamma_k=0
	\end{array}
	\right. .\label{eta}
\end{align}
Using Eq. (\ref{error}) and Eq. (\ref{hold}), $P_k$ can be written as a function of $\eta_k$ as
\begin{align}
P_k = f^{\eta_k}(\bar{P}).\nonumber
\end{align}
Here $f^n(\cdot)$ denotes the $n$-th fold composition of $f(\cdot)$ with $f^0(X)=X$. As proved in \cite{shi2012scheduling}, all possible values of $P_k$ form a totally ordered set $\mathbb{S}$, i.e.
\begin{align}
\mathbb{S} \triangleq \{ \overline{P}, f(\overline{P}), f^2 (\overline{P}),f^3 (\overline{P}), \dots\}\nonumber
\end{align}
and
\begin{align}
\bar{P}\leq f(\bar{P})\leq f^2(\bar{P})\leq \dots.\label{p_mono}
\end{align}
After receiving the packet, the estimator will send an acknowledgment (ACK) back to the sensor. This enables the sensor to use information about $P_k$ to make scheduling decisions. The ACK channel is often assumed to be reliable\cite{leong2018transmission,ding2020remote}. 
We go beyond that by taking into account the packet dropouts across this channel. In the following sections, we will show that the reliability of the ACK channel plays a very crucial role in system privacy.
Denote $\lambda_a$ as the successful reception probability of the ACK signal. Let $\gamma^a_k$ be an indicator variable such that
\begin{align}
\gamma^a_k=\left\{\begin{array}{cc}
1 & \text{ACK is received by the sensor}\\
0 & \text{a packet dropout occurs}
\end{array}\right. .\nonumber
\end{align}
Then, the conditional distribution of $\gamma_k^a$ is given as
\begin{equation}
\begin{aligned}
\mathbb{P}(\gamma^a_k=1\mid \gamma_k=1) = & \lambda_a\\
\mathbb{P}(\gamma^a_k=0\mid \gamma_k=1) = & 1-\lambda_a\\
\mathbb{P}(\gamma^a_k=0\mid \gamma_k=0) = & 1.
\end{aligned} \label{gamma}
\end{equation}

\subsection{Problem setup}
 We consider a scenario that an adversary  can overhear the transmissions and hack the sensor  to re-schedule the transmissions. 
 Let $\gamma_k^e$ be an indicator of a successful reception at the adversary, such that $\gamma_k^e=1$ if the transmission is overheard by the adversary and $\gamma_k^e=0$ otherwise. Assume the successful reception probability is $\lambda_e$. We have
\begin{equation}
\mathbb{P}(\gamma_k^e=1\mid \nu_k=1) = \lambda_e.\label{trane}
\end{equation}

\par Similar to Eq. (\ref{estimate}) and Eq. (\ref{error}), assuming that the processes $\gamma_k^e$ and $\gamma_k$ are mutually independent, the adversary's optimal estimate is in the form of 
\begin{align}
\hat{x}_{k}^e  = & \left\{\begin{array}{cl} \hat{x}_{ k}^s, & 
\gamma_{k}^e \nu_k= 1 \\ A \hat{x}_{k-1}^e, & 	\gamma_{k}^e \nu_k= 0, 
\end{array}  \right. \label{eav:state} \\
P_{k}^e  = & \left\{\begin{array}{cl}  \overline{P}  &   \gamma_{k}^e \nu_k=  1 \\ 
f(P_{k-1}^e),  &
\gamma_{k}^e \nu_k=  0 \end{array} \right. \label{eav:error} ,
\end{align}
with $\hat{x}_k^e\triangleq \mathbb{E}(x_k\mid \mathcal{I}_k^e)$,  $P_k^e=\mathbb{E}[(x_k-\hat{x}_k)(x_k-\hat{x}_k)^\top\mid \mathcal{I}_k^e]$, and 
$\mathcal{I}_k\triangleq\{\nu_0\gamma_0^e,\dots,\nu_k\gamma_k^e, \nu_0\gamma_0^e\hat{x}_0^s,\dots,\nu_k\gamma_k^e\hat{x}_k^s\}.$

As depicted in Fig. \ref{setup}, after launching the attack to the sensor, the adversary sends an acknowledgment signal back to the sensor such that the holding time at the adversary, denoted as $\eta_k^e$, can be inferred from the ACKs according to (\ref{eta_e}),
\begin{align}
\eta_k^e=\left\{\begin{array}{cc}
0 & \nu_k\gamma_k^e=1\\
\eta_{k-1}^e+1 & \nu_k\gamma_k^e=0
\end{array}
\right. ,\label{eta_e}
\end{align}
such that the transmissions can be re-scheduled with $\eta_k^e$ taken into account. We assume the ACK channel between the adversary and the sensor has no packet dropouts. In this case, $\eta_k^e$ is exactly known by the sensor, which is the most advantageous scenario to the adversary. Studying this scenario will help us analyze the  maximal information leakage to the adversary. 

We evaluate the estimation performance with an average of the expected error covariance (AEEC) defined as
\begin{align}
J_l=& \lim_{T\rightarrow\infty}\frac{1}{T} \sum_{k=1}^T\text{tr}(\mathbb{E}(P_k)), \nonumber\\
J_e = & \lim_{T\rightarrow\infty}\frac{1}{T} \sum_{k=1}^T\text{tr}(\mathbb{E}(P_k^e)), \nonumber
\end{align}
where $J_l$ corresponds to the legitimate estimator and $J_e$ corresponds to the adversary. In line with \cite{tsiamis2017state}, we study "perfect secrecy" at the legitimate estimator, here defined as keeping $J_l$ bounded at the legitimate estimator and $J_e$ unbounded at the adversary.

To study perfect secrecy, we will next elucidate the optimal design of the malicious transmission policy for the adversary as well as the reference transmission policy adopted by the sensor without an intrusion in the following section. 

\section{Remote state estimation without an intrusion}
\label{sec3}
Before exploring the malicious transmission policy, we first look at the estimation performance at the adversary without intruding on the sensor. This will provide a baseline for the design of the malicious policy. 

In remote state estimation, event-based transmission strategies are considered to be more efficient in saving energy \cite{trimpe2014event} and preserving privacy \cite{leong2018transmission}. In particular, the transmission command $\nu_k$ often depends on the sensor's belief about the remote estimator's expected error covariance $P_{k-1}$ (which is equivalent to the belief about the holding time $\eta_{k-1}$).
Denote this belief as $\eta_{k-1}^s$. If the ACK channel is reliable, i.e. $\lambda_a=1$, $\eta_k$ can be accurately inferred from $\{\gamma_0^a,\dots,\gamma_k^a\}$ according to Eq. (\ref{gamma}), i.e. $\eta_k^s=\eta_k$. The sensor can then schedule the transmissions according to 

		\textbf{Policy 1:}
\begin{equation}
\mathbb{P}(\nu_{k+1}=1 \mid \eta_{k}^s=i) = \tau_i,\label{policy}
\end{equation}
with $0\leq \tau_i\leq 1$. Policy 1 is referred as a reference policy for the sensor. Different from \cite{leong2018transmission} and \cite{trimpe2014event}, where the transmission probability $\tau_i$ is either $0$ or $1$, we consider a randomized covariance-based transmission policy as given in \cite{li2017randomized} for larger design space.

If $\lambda_a<1$, as discussed in \cite{ding2018attacks}, a belief of $\eta_k$ can be estimated with Bayesian methods. Since the number of belief states exponentially increases with the time since the last $\gamma_k^a=1$, this approach will dramatically complicate the design. An easy though suboptimal approach for the sensor is to directly update $\eta_k^s$ according to the most recent acknowledgement as
\begin{equation}
\eta_k^s = \left\{\begin{array}{cc}
0 & \gamma_k^a = 1\\
\eta_{k-1}^s+1 & \gamma_k^a = 0
\end{array}\right..\label{eta_s}
\end{equation}  
 Note that our conclusions in the latter sections are validated if other estimation approaches are applied to $\eta_k^s$. For an ease in presentation, we assume that $\eta_k^s$ evolutes according to Eq. (\ref{eta_s}) in the following.

In view of Eq. (\ref{trans1}) and (\ref{eta}), it can be verified that the sensor's holding time $\eta_k^s$ satisfies a Markovian property since the transition probability from $\eta_k^s$ to $\eta_{k+1}^s$ solely  depends on $\eta_k^s$. This enables us to formulate the evolution of $\eta_k^s$ as a Markov chain. According to this, we analyze the stationary distribution of the expected error covariance at the legitimate estimator and the adversary.

As described in Fig. 3, we define the sensor's holding time $\eta_k^s$ as the state of the Markov chain. The value of the state is taken from the set $\mathbb{S}=\{0,1,\dots,\}$. The transition probabilities can be derived from Eq. (\ref{trans1}) and Eq. (\ref{gamma}) as
\begin{equation}
\begin{aligned}
\mathbb{P}(\eta^s_{k+1}=i+1\mid \eta^s_k=i) = & 1-\tau_i\lambda\lambda_a\\
\mathbb{P}(\eta^s_{k+1}=0\mid \eta^s_k=i) = & \tau_i\lambda\lambda_a.
\end{aligned} \label{trans4}
\end{equation}
\begin{figure*}[h]
	\centering
	\label{markov_2}
	\begin{tikzpicture}[->, auto, semithick, node distance=2.5cm]
	\tikzstyle{every state}=[fill=white,draw=black,thick,text=black,scale=1]
	\node[state,minimum size=0.8cm]    (A)                     {\footnotesize $0$};
	\node[state,minimum size=0.8cm]    (B)[right of=A]   {\footnotesize$1$};
	\node[state,minimum size=0.8cm]    (C)[right of=B]   {\footnotesize$2$};
	\node[draw=none,right=of C]           (C-D) {$\dots$};
	\node[state,minimum size=0.8cm]    (D)[right of={C-D}]   {\footnotesize$i$};
	\node[draw=none,right=of D]           (F) {$\dots$};
	\path
	(A) edge   node{\footnotesize$1-\tau_0\lambda_a\lambda$}      (B)
	(B) edge                node{\footnotesize$1-\tau_1\lambda_a\lambda$}           (C)
	(C) edge                node{\footnotesize$1-\tau_2\lambda_a\lambda$}           (C-D)
	(C-D) edge                node{\footnotesize$1-\tau_{i-1}\lambda_a\lambda$}           (D)
	(D) edge                node{\footnotesize$1-\tau_{i}\lambda_a\lambda$}           (F)
	(B)	edge[bend left,above]  node{\footnotesize$\tau_1\lambda_a\lambda$} (A)
	(C)	edge[bend left,above]  node{\footnotesize$\tau_2\lambda_a\lambda$} (A)
	(D)	edge[bend left,above]  node{\footnotesize$\tau_i\lambda_a\lambda$} (A)
	(A) edge[loop] node{\footnotesize$\tau_0\lambda_a\lambda$} (A)
	;
	\end{tikzpicture}
	\caption{Markov chain: $\mathbb{P}(\gamma_k|\nu_k=1)=\lambda$. $\mathbb{P}(\gamma^a_k|\gamma_k=1)=\lambda_a$. The state $s_k=i$ denotes the event $\eta_k^s=i$.}
\end{figure*}
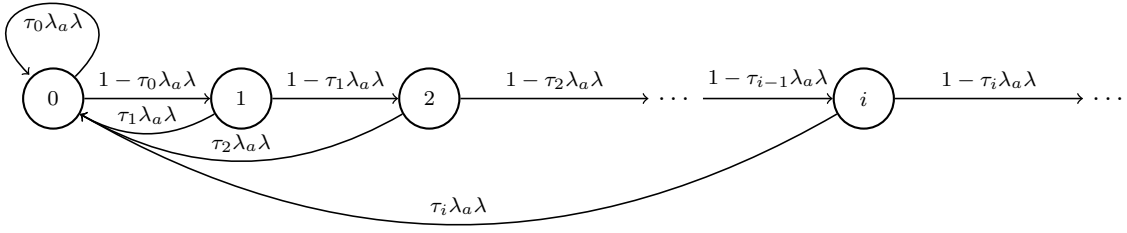
Since the transition probabilities are independent of time, the Markov chain is time-homogeneous. Denote $\pi_i^s$ as the stationary probability of the state $\eta_k^s=i$. The unique stationary distribution of $\eta_k^s$ follows \cite{li2017randomized}
\begin{equation}
\begin{aligned}
\pi_0^s = &\frac{1}{1+\sum_{l=0}^\infty \prod_{m=0}^l(1-\tau_m\lambda\lambda_a)} ,\\
\pi_i^s = &\frac{\prod_{m=0}^{i-1}(1-\tau_m\lambda\lambda_a)}{1+\sum_{l=0}^\infty \prod_{m=0}^l(1-\tau_m\lambda\lambda_a)},~i\geq 1.
\end{aligned} 	\label{pi_e}
\end{equation}
Define 
\begin{equation}
J_u=\sum_{j=0}^\infty \pi_j^s \text{tr}(f^j(\bar{P})).\label{J_u}
\end{equation}
In Proposition \ref{p1}, we show that $J_u$ is an upper bound on the AEEC at the legitimate estimator. 

\begin{prop}
	\label{p1}
	For any given transmission policy parameterized by the transmission probability $\{\tau_i\},~i\in\mathbb{N}$, the estimator's averaged expected error covariance $J_l$ is upper bounded as
	\begin{equation}
	J_l\leq J_u. \label{error_upper}
	\end{equation}
	The equality is achieved when $\lambda_a=1$.
\end{prop}
\begin{proof}
	In view of Eq. (\ref{eta}), (\ref{gamma}) and (\ref{eta_s}), we have 
	\begin{equation}
	\eta_k\leq \eta_k^s.\nonumber
	\end{equation} 
	Since $P_k=f^{\eta_k}(\bar{P})$ and the function $f^n(\bar{P})$ is non-decreasing with $n$, it can be derived that
	\begin{align}
	J_l=&\lim_{T\rightarrow\infty}\frac{1}{T} \sum_{k=1}^T\text{tr}(\mathbb{E}(P_k)) =  \lim_{T\rightarrow\infty}\frac{1}{T} \sum_{k=1}^T\text{tr}(\mathbb{E}(f^{\eta_k}(\bar{P}))) \nonumber\\
	\leq & \lim_{T\rightarrow\infty}\frac{1}{T} \sum_{k=1}^T\text{tr}(\mathbb{E}(f^{\eta^s_k}(\bar{P}))) = \sum_{j=0}^\infty \pi_j^s \text{tr}(f^j(\bar{P}))=J_u.\nonumber
	\end{align}
	If $\lambda_a=1$, we have $\eta_k=\eta_k^s$ and the equality is achieved.
\end{proof}	
Since the sensor can hardly infer the exact value of $\eta_k$ when $\lambda_a<1$, it can take $J_u$, which depends on $\eta_k^s$, as an index of the legitimate estimator's performance. Sufficient and necessary conditions for a bounded $J_u$ are given in Proposition \ref{prop_ns} based on Lemma \ref{prop_ul}.
\begin{lem}
	\label{prop_ul}
	For any positive integer $n$, the function $\text{tr}f^n(\bar{P})$ is upper and lower bounded by a function of $\bar{P}$ and $n$ as
	\begin{equation}
	B_l(\bar{P},n)	\leq	\text{tr}f^n(\bar{P}) \leq B_u(\bar{P},n) \label{f_bound}
	\end{equation}
	with 
	\begin{equation}
	\begin{aligned}
	& B_u(\bar{P},n) =  n_s\sigma_{\max}(\bar{P})\sigma_{\max}^{2}(A^n) +n_s\sigma_{\max}(Q)\sum_{i=0}^{n-1}\sigma_{\max}^{2}(A^i),
	\\
	& B_l(\bar{P},n) =\lambda_{\min}(\bar{P})\lambda_{\max}\left((A^\top)^nA^n\right) \nonumber\\ & \quad\quad\quad\quad\quad+\lambda_{\min}(Q)\sum_{i=0}^{n-1}\lambda_{\max}\left((A^\top)^iA^i\right).
	\end{aligned} \nonumber
	\end{equation}
	Here $n_s$ is the dimension of the system (see Eq. (\ref{sys})). 
\end{lem}
\begin{proof}
	See Appendix \ref{prlemm2}.
\end{proof}	
\begin{prop}
	\label{prop_ns}
The legitimate estimator's performance index $J_u$ is bounded if there exists an integer $N_r>0$ such that when $i\geq N_r$, the transmission probability $\tau_i$ satisfies that \begin{equation}
\tau_i>\frac{1}{\lambda\lambda_a}\left(1-\frac{1}{|\lambda_{\max}(A)|^2}\right).\label{sufficient}
\end{equation}
 Meanwhile, if $J_u$ is bounded and $|\lambda_{max}(A)|\geq 1$, then $\tau_i$ must satisfy that 
 	\begin{align}
 \lim_{l\rightarrow\infty} \prod_{i=0}^{l-1}(1-\tau_i\lambda)= 0. \label{lim}
 \end{align}
\end{prop}
The proof of Proposition \ref{prop_ns} is given in Appendix \ref{ptoP2}. Next, we will investigate, without specifying the transmission policy, the estimation performance at the adversary on the premise that $J_u$ is bounded. In view of Lemma \ref{prop_ul}, the adversary's AEEC $J_e$ is always bounded if the system is strictly stable, i.e. $|\lambda_{\max}(A)|<1$. Hence, in Proposition \ref{prop0}, we exclude this trivial case and focus on marginally stable and unstable systems. 
\begin{prop} \label{prop0}
	If a transmission policy parameterized by $\{\tau_i\},~i\in\mathbb{N}$ can make $J_u$ bounded, then the adversary's AEEC $J_e$ has the following properties:
	\begin{itemize}
		\item[-i)]
		If the system is marginally stable, i.e. $\lambda_{\max}(A)=1$,  $J_e$ is bounded, but can be made arbitrarily large. Namely, given any scalar $\underline{b}>0$, there exists a transmission policy $\{\tau_i\}$, $i\in\mathbb{N}$, such that
		\begin{align}
		J_e >\underline{b}.\nonumber
		\end{align}
		\item[-ii)] If $0<\lambda_e<1$ and the system is unstable, i.e. $|\lambda_{\max}(A)|>1$, then
		there exists a transmission policy ensuring perfect secrecy, i.e. $J_e=\infty$.
	\end{itemize}
\end{prop}	

\begin{proof}
	See Appendix \ref{ptprop3}.
\end{proof}
Proposition \ref{prop0} suggests that by properly designing the transmission policy, the sensor is capable of driving the AEEC at the adversary to infinity if the dynamics in Eq. (\ref{sys}) is unstable, and arbitrarily large if the system is marginally stable. This motivates the adversary to secretly re-schedule the transmissions  to obtain a more accurate estimate.

\section{Stealthy malicious transmission policy}
\label{sec4}
In this section, we will adopt the adversary's perspective and study how to optimize the malicious transmission policy such that it is more advantageous to the adversary while staying undetected. 
We will first model the transmission process as a Markov decision process, and then devise a stealthy constraint in terms of the state-action pairs associated with the MDP. 
After that, the optimal malicious policy will be derived by solving a constrained Markov decision process. 

\subsection{Markov decision process}
\label{formulation}
As shown in Fig. \ref{setup}, the sensor is controlled by the adversary after the hacking, and both $\eta_k^e$ and $\eta_k^s$ are known to the sensor. This enables the adversary to modify the transmission policy in (\ref{policy}) by taking into account the adversary's covariance $P_k^e$ (equivalent to $\eta_k^e$) in the general form of 
\begin{align}
\mathbb{P}(\nu_{k+1}=1 \mid \eta_{k}^s=i, \eta_{k}^e=j) = \tilde{\tau}_{ij},~\forall i,j\in \mathbb{N}.\label{policy2}
\end{align}
Note from Eq. (\ref{eta_e}) and (\ref{eta_s}) that both $\eta_k^s$ and $\eta_k^e$ are Markovian. We can formulate a Markov decision process to derive the optimal malicious transmission policy for the adversary. Define the pair $(\eta_k^s,\eta_k^e)$ as the state $S_k$ and the action as the transmission command $\nu_k\in\mathbb{A}$ with $\mathbb{A}=\{0,1\}$. Based on Eq. (\ref{trans1}), (\ref{eta}), (\ref{trane}), the state transition probabilities for $S_k=(i,j)$ are derived as follows
\begin{equation}
\begin{aligned}
& \mathbb{P}(S_{k+1}\mid S_k,\nu_{k+1})=\\
&\left\{
\begin{array}{ll}
1 & S_{k+1}=(i+1,j+1),~\nu_{k+1}=0 \\
(1-\lambda\lambda_a)(1-\lambda_e) & S_{k+1}=(i+1,j+1),~\nu_{k+1}=1\\
\lambda\lambda_a(1-\lambda_e) & S_{k+1}=(0,j+1),~\nu_{k+1}=1\\
(1-\lambda\lambda_a)\lambda_e& S_{k+1}=(i+1,0),~\nu_{k+1}=1\\
\lambda\lambda_a\lambda_e & S_{k+1}=(0,0),~\nu_{k+1}=1
\end{array}
\right.
\end{aligned}.
\label{trans6}
\end{equation}	 
In Problem \ref{prob1}, the action $\nu_k$ is determined by minimizing the AEEC at the adversary while satisfying a stealthy constraint. 
\begin{prob}
	\label{prob1}
\begin{align}
\min_{\nu_k} ~~&~~J_e=\lim_{T\rightarrow\infty}\frac{1}{T} \sum_{k=1}^T\text{tr}(\mathbb{E}(P_k^e))\nonumber\\
\text{s.t.} ~~~&~~\text{stealthy constraint}.\nonumber
\end{align}
\end{prob}
Next, we will detail the design of the stealthy constraint.

\subsection{Stealthy constraint}
From Eq. (\ref{policy2}), the transmission probability at $\eta_k^s$ follows
\begin{equation}
\mathbb{P}(\nu_{k+1}=1 \mid \eta_{k}^s=i) = \sum_{j=0}^\infty \tilde{\tau}_{ij}\mathbb{P}(\eta_{k}^e=j) ,\nonumber
\end{equation}
which might be significantly differ from the reference policy in Eq. (\ref{policy}). 
Being aware of the potential existence of the adversary, the legitimate estimator will proactively detect whether the transmission policy in (\ref{policy}) is modified or not based on its information set $\mathcal{I}_k$. This is essentially a parameter estimation problem for a hidden Markov model (HMM) \cite{molloy2018minimax} if we take $\eta_k^s$ as the hidden states and $\nu_k\gamma_k$ as the observations. The state transition probability is given in (\ref{trans4}) and the emission probabilities can be derived from Eqs. (\ref{trans1}), (\ref{gamma}), (\ref{policy}), and (\ref{eta_s}), which follow
\begin{equation}
\begin{aligned}
&\mathbb{P}(\nu_{k}\gamma_{k}=1\mid \eta_k^s=i+1) 
= \frac{\mathbb{P}(\nu_{k}\gamma_{k}=1\mid \eta_{k-1}^s=i)}{\mathbb{P}(\eta_k^s=i+1\mid \eta_{k-1}^s=i)}\\
&\quad\quad\quad\quad\quad\quad\quad\quad\quad\quad~=\frac{\tau_i\lambda}{1-\tau_i\lambda\lambda_a}\\
&\mathbb{P}(\nu_{k}\gamma_{k}=0\mid \eta_k^s=i+1) =1-\frac{\tau_i\lambda}{1-\tau_i\lambda\lambda_a}\\
&\mathbb{P}(\nu_{k}\gamma_{k}=1\mid \eta_k^s=0) =1\\
&\mathbb{P}(\nu_{k}\gamma_{k}=0\mid \eta_k^s=0) =0
\end{aligned}. \label{emission}
\end{equation} 
Both the transition probability and emission probability are determined by the transmission policy $\{\tau_i\}$, and thus the legitimate estimator can detect variations in the parameters of this HMM by observing the indicator variables $\nu_0\gamma_0,\dots,\nu_k\gamma_k$. This HMM has infinite hidden states since $\eta_k^s$ takes value from the set $\mathbb{N}$. It can be equivalently represented by a finite-state HMM by aggregating the state $\eta_k^s\geq N_r$ for any positive integer $N_r$. From Theorem 6 in \cite{allman2009identifiability}, the transmission probability ${\tau_i}$ with $i\in[0,N_r]$ is generically identifiable from the the marginal distribution of $N_r$ consecutive variables \begin{equation}o_{k+1:k+N_r}=(\nu_{k+1}\gamma_{k+1},\dots,\nu_{k+N_r}\gamma_{k+N_r}). \nonumber
\end{equation}
 This property motivates the legitimate estimator to formulate a Hypothesis test in the form of



\noindent\textbf{Hypothesis test:}
\begin{itemize}
	\item[$\mathbf{H}_0$:] The marginal distribution of  $o_{k+1:k+N_r}$ follows the distribution $\mathcal{P}_0$;
	\item[$\mathbf{H}_1$:] The marginal distribution of  $o_{k+1:k+N_r}$ does not follow the distribution $\mathcal{P}_0$;
\end{itemize}
Denote the distribution of $o_{k+1:k+N_r}$ after the intrusion as $\mathcal{P}_1$. To stay undetected, the adversary intends to have the observations $o_{k+1:k+N_r}$ follow $\mathcal{P}_1$ while keeping Hypothesis $H_0$ not rejected by the legitimate estimator. Consider that the Kullback-Leibler distance $	\mathcal{KL}(\mathcal{P}_1||\mathcal{P}_0)$, defined as 	\begin{equation}
\mathcal{KL}(\mathcal{P}_1||\mathcal{P}_0)=\sum_{x\in\Omega_s}\mathcal{P}_1(x)\log\frac{\mathcal{P}_1(x)}{\mathcal{P}_0(x)},\nonumber
\end{equation}
gives the expected log-likelihood with which Hypothesis $H_0$ can be rejected per event \cite{press1992numerical}. We formulate a stealthy constraint as  
	\begin{equation}
\mathcal{KL}(\mathcal{P}_1||\mathcal{P}_0)\leq\epsilon_{kl},\label{stealthy_kl}
\end{equation}
with $\epsilon_{kl}>0$ denoting a stealthy tolerance. A small $\epsilon_{kl}$ will make the intrusion less likely being detected. 



 Next, we will study the effects of the detection horizon $N_r$, the stealthy tolerance $\epsilon_{kl}$, the successful reception probability $\lambda_a$ and the reference policy (\ref{policy}) onto perfect secrecy. In Theorem \ref{safe}, we show that if the ACK channel is reliable, i.e. $\lambda_a=1$, and $N_r$ is large enough, then perfect secrecy can be achieved for unstable systems.
\begin{thm} \textbf{(Attainability of ``perfect secrecy")}:
Given that $\lambda_a=1$, $0<\lambda_e<1$ and $|\lambda_{\max}(A)|>1$, devise the reference transmission policy as
	\begin{align}
\mathbb{P}[\nu_{k+1}=1\mid \eta_k^s]=\left\{\begin{array}{cc}
1 & \eta_k^s > \bar{t}\\
0 & \text{else}
\end{array}\right. \label{candidate}
\end{align}
with $\bar{t}$ satisfying 
 \begin{align}
\frac{1}{| \lambda_{\max}(A)|^{2(\bar{t}+1)}}< \lambda(1-\lambda_e).\nonumber
\end{align}
If $N_r>\bar{t}$, there is no malicious transmission policy in the form of Eq. (\ref{policy2}) that can simultaneously make 
 $J_e$ bounded and fulfill the stealthy constraint in Eq. (\ref{stealthy_kl}). \label{safe}
\end{thm}
\begin{proof}
See Appendix \ref{pr_th6}.
\end{proof}
In the following sections, we will show an opposite conclusion when there are packet dropouts in the ACK channel. In this case, the state $\eta_k$ and $\eta_k^s$ cannot be synchronized, making perfect secrecy not attainable for any reference policy and any detection horizon.

\subsection{Constrained MDP for an unreliable ACK channel} \label{cmdp}

In this section, we will transform the stealthy constraint in (\ref{stealthy_kl}) into a $l_1$-norm constraint on the stationary distribution of a state-action pair, such that Problem \ref{prob1} can be cast into a linear program. Based on this, we prove the non-attainability of perfect secrecy by studying the feasibility of the constrained MDP. 
\begin{prop}
	\label{al_stealthy}
 Define $\omega_s(i,a)$ as the stationary probability of the state-action pair $(\eta_k^s=i,\nu_{k+1}=a)$ before the intrusion and $\rho_s(i,a)$ as that after the intrusion. Given that $\lambda_a<1$, for any stealthy tolerance $\epsilon_{kl}>0$ and any detection horizon $N_r>0$, there exists a scalar $\epsilon_s>0$ such that if 
	\begin{align}
\|\omega_s-\rho_s\|_1\leq \epsilon_s,\label{stealth_refine}
	\end{align}
	Eq. (\ref{stealthy_kl}) holds.
\end{prop}
\begin{proof}
See Appendix \ref{proof_pi}.
\end{proof}
According to Proposition \ref{al_stealthy}, we can take (\ref{stealth_refine}) as the stealthy constraint in Problem \ref{prob1}. Define a tensor $\omega$ such that its component $\omega(i,j,a)$ is the stationary probability of the appearance of the state-action pair $(S_k=(i,j),\nu_{k+1}=a)$ with $i,j\in\mathbb{N}, a\in\mathbb{A}$ before the launch of the intrusion. Similarly, we can define $\rho$ as the stationary distribution after the intrusion. Based on $\omega$ and $\rho$, we will show that Problem 1 is equivalent to a linear program. 

Let's look at $\rho$ first. Note from \cite{altman1999constrained} that $\rho$ is referred to as an occupation measure in constrained MDP, which satisfies 
\begin{equation}
\begin{aligned}
&\sum_{a\in\mathbb{A}} \rho(i,j,a) - \sum_{i',j'\in\mathbb{N},a\in\mathbb{A}} \rho(i,j,a) \mathcal{P}_{i',j',a,i,j} = 0,\\
&\sum_{i,j\in\mathbb{N},a\in\mathbb{A}}\rho(i,j,a) = 1\\
&  \rho(i,j,a) \geq 0,~\forall i,j\in\mathbb{N} ~\forall a\in\mathbb{A}.
\end{aligned}\label{ce3} 
\end{equation}
Here $\mathcal{P}_{i',j',a,i,j}$ is a short form of the state transition probability $\mathbb{P}(S_{k+1}=(i,j)\mid S_k=(i',j'),\nu_{k+1}=a)$ in Eq. (\ref{trans6}). For ease of notation, we write Eq. (\ref{ce3}) in a compact form as $\rho \in \Omega$.  Moreover, we can express $J_e$ and $\rho_s$ as linear combinations of $\rho$:
\begin{align}
J_e
= &\lim_{n\rightarrow\infty}\sum_{i\in\mathbb{N}}\sum_{a\in\mathbb{A}}\left(\sum_{j\leq n}\rho(i,j,a)\text{tr}(f^j(\bar{P}))\right),\label{linear}\\
\rho_s(i,a) =& \sum_{j\in\mathbb{N}} \rho(i,j,a).\nonumber
\end{align}
Similarly, we have $\omega\in\Omega$ and $\omega_s(i,a) = \sum_{j\in\mathbb{N}} \omega(i,j,a).$
In addition, to be accordance with Eq. (\ref{policy}), $\omega$ should satisfy that
\begin{equation}
\tau_i = \frac{\omega(i,j,a=1)}{\sum_{a\in\mathbb{A}}\omega(i,j,a)},~\forall i,j\in\mathbb{N}.\nonumber
\end{equation}
If the transmission policy $\{\tau_i\}$ is fixed, the resulted Markov chain is ergodic and thus has a unique stationary distribution, indicating that $\omega$, as well as $\omega_s$ are uniquely determined. Accordingly, the optimal malicious policy can be derived from Problem 	\ref{prob2} where the occupation measure $\rho$ is taken as a decision variable. 

\begin{prob}
	\label{prob2}
\begin{align}
\min_{\mathbf{\rho}} ~&~J_e\nonumber\\
\text{s.t.}~~&~  \rho\in\Omega \nonumber\\
&~ \sum_{i\in\mathbb{N},a\in\mathbb{A}}|\sum_{j\in\mathbb{N}} \rho(i,j,a)-\omega_s(i,a)|\leq \epsilon_s.\label{stealthy}
\end{align}
\end{prob}
After introducing auxiliary variables $\epsilon_{i,a}$ with $i\in\mathbb{N},a\in\mathbb{A}$, Eq. (\ref{stealthy}) can be further simplified into a set of linear equalities as
\begin{equation}
\begin{aligned}
&-\epsilon_{i,a}\leq \sum_{j\in\mathbb{N}} \rho(i,j,a)-\omega_s(i,a)\leq \epsilon_{i,a}\\
& \sum_{i\in\mathbb{N},a\in\mathbb{A}} \epsilon_{i,a} \leq \epsilon_s
\end{aligned}.\label{linearity}
\end{equation}
Therefore, Problem \ref{prob2} is a linear program. 

If Problem 	\ref{prob2} is feasible, which means there is
a solution that can make $J_e$ bounded and the constraints satisfied,  for any $\epsilon_s$ and any reference policy $\{\tau_i\}$, then perfect secrecy is not attainable. Thus, the feasibility of this constrained linear program determines the attainability of perfect secrecy. In the following section, we will prove this by truncating the original MDP into a finite-state MDP. The truncation will also be employed in the practical implementation of the malicious policy. 

\subsection{Finite dimensional approximation}
\label{fa}
We truncate the MDP with  $(N_t+1)^2$ states reserved by defining the state $S_k$ as
\begin{equation}
S_k=\left\{\begin{array}{ll}
(\eta_k^s,\eta_k^e) &  \eta_k^s,\eta_k^e<N_t \\
(N_t,\eta_k^e ) &  \eta_k^s\geq N_t,~\eta_k^e<N_t\\
(\eta_k^s,N_t) & \eta_k^s<N_t,~\eta_k^e\geq N_t\\
(N_t,N_t) & \eta_k^s,\eta_k^e\geq N_t
\end{array}\right..\nonumber
\end{equation}
Here $N_t$ is the truncation horizon. The transition probability corresponding to the truncated MDP is revised as
 {\small
	\begin{equation}
	\begin{aligned}
	& \mathbb{P}(S_{k+1}\mid S_k,\nu_{k+1})=\\
	&\left\{
	\begin{array}{lc}
	1 & \begin{array}{c}  S_{k+1}=(\min(i+1,N_t),\min(j+1,N_t))\\
	\nu_{k+1}=0 \end{array}\\
	(1-\lambda\lambda_a)(1-\lambda_e) & \begin{array}{c}
	S_{k+1}=(\min(i+1,N_t),\min(j+1,N_t))\\
	\nu_{k+1}=1 \end{array}\\
	\lambda\lambda_a(1-\lambda_e) & \begin{array}{c}
	S_{k+1}=(0,\min(j+1,N_t))\\
	\nu_{k+1}=1\end{array}\\
	(1-\lambda\lambda_a)\lambda_e& \begin{array}{c} S_{k+1}=(\min(i+1,N_t),0)\\
	\nu_{k+1}=1 \end{array}\\
	\lambda\lambda_a\lambda_e & S_{k+1}=(0,0),~\nu_{k+1}=1
	\end{array}
	\right. .
	\end{aligned}
	\label{trans8}
	\end{equation}	}
\normalsize
Define a tensor $\rho_t\in\mathbb{R}^{(N_t+1)\times (N_t+1)\times 2}$ as the occupation measures for the states of the  truncated MDP. Similar to Eq. (\ref{ce3}), we have $\rho_t$ satisfy that 
\begin{equation}
\begin{aligned}
&\sum_{a\in\mathbb{A}} \rho_t(i,j,a) - \sum_{0\leq i',j'\leq N_t,a\in\mathbb{A}} \rho_t(i,j,a) \mathcal{P}_{i',j',a,i,j} = 0,\\
&\sum_{0\leq i,j\leq N_t,a\in\mathbb{A}}\rho_t(i,j,a) = 1\\
&  0\leq \rho_t(i,j,a) \leq 1,~\forall ~0\leq i,j\leq N_t~\forall a\in\mathbb{A}.
\end{aligned}\label{ce4} 
\end{equation}
Here $\mathcal{P}_{i',j',a,i,j}$ denotes the transition probability in (\ref{trans8}). Denote Eq. (\ref{ce4}) as $\rho_t\in\Omega_{N_t}$ for simplicity in notation. Meanwhile, we limit the number of free variables in the malicious policy as  
\begin{align}
\tau_{ij}^t = \left\{\begin{array}{ll}
\tilde{\tau}_{ij} & 0 \leq i,j< N_t \\
1 & i\geq N_t ~\text{or}~ j\geq N_t
\end{array}\right., \label{policy_trun}
\end{align}
where $\tilde{\tau}_{i,j}$ with $0\leq i,j\leq N_t$ are taken as free variables to be determined and the others are fixed to be $1$. In accordance to (\ref{policy_trun}), $\frac{\rho_t(i,j,a=1)}{\sum_{a\in\mathbb{A}}\rho_t(i,j,a)}=1$ when $i=N_t$ or $j=N_t$, which gives
\begin{align}
& \rho_t(i,j,a=0)=0,\quad i=N_t~\text{or}~j= N_t.\label{rhot21}
\end{align} 
With Proposition \ref{prop_8}, we show that $J_e$ is a linear function of $\rho_t$ and the policy (\ref{policy_trun}) makes $J_e$ bounded for any $N_t$.
\begin{prop}
	The averaged expected error covariance corresponding to the truncated policy in Eq. (\ref{policy_trun}) can be expressed in terms of $\rho_t$ as
	\begin{align}
	J_e
	= &\sum_{i,j\leq N_t,~a\in\mathbb{A}}\rho_t(i,j,a)\text{tr}(f^j(\bar{P}))+\nonumber\\
	&\sum_{i\leq N_t}\rho_t(i,N_t,a=1)\sum_{j=0}^\infty\lambda_e(1-\lambda_e)^{j-N_t}\text{tr}f^j(\bar{P}), \label{J_e_finite}
	\end{align}
	which is bounded for any $N_t\geq 0$ if $\lambda_e\in(1-\frac{1}{|\lambda_{max}(A)|^2},1)$.\label{prop_8}
\end{prop}
\begin{proof}
	See Appendix \ref{proof_prop8}.
\end{proof}	
This enables us to devise a linear program based on $\rho_t$ as

\begin{prob}
	\label{prob3}
\begin{align}
\min_{\mathbf{\rho}_t,\epsilon_{i,a}} &~J_e\nonumber\\
\text{s.t.}~&~ \rho_t\in\Omega_{N_t},~(\ref{rhot21}) \nonumber\\
&~\text{(Stealthy constraint):}\nonumber\\
&-\epsilon_{i,a}\leq \sum_{0\leq j\leq N_t}\rho_t(i,j,a)-\omega_s(i,a)\leq \epsilon_{i,a}, 0\leq i< {N_t}\label{s1}\\
&-\epsilon_{N_t,a}\leq \sum_{0\leq j\leq N_t}\rho_t(N_t,j,a)-\sum_{N_t}^\infty\omega_s(i,a)\leq \epsilon_{N_t,a}\nonumber\\
&\sum_{0\leq i\leq {N_t},~a\in\mathbb{A}} \epsilon_{i,a} \leq \epsilon_s.\label{s2}
\end{align}
\end{prob}
Feasibility of Problem \ref{prob3} is analyzed in the following theorem.
\begin{thm} \textbf{(Feasibility)}:
	Given that $\lambda_a<1$ and $\lambda_e\in(1-\frac{1}{|\lambda_{max}(A)|^2},1)$, there exists an integer $N_t$ such that Problem \ref{prob3} is ensured to be feasible for
	any stealthy tolerance $\epsilon_s>0$. \label{feasibility}
\end{thm}
\begin{proof}
	See Appendix \ref{prtoThm}.

Denote the solution of Problem \ref{prob3} as $\rho_t^c$. The transmission probability $\tilde{\tau}_{ij}$ in (\ref{policy_trun}) can be determined as
\begin{align}
\tilde{\tau}_{i,j} = & \frac{\rho_t^c(i,j,a=1)}{\sum_{a\in\mathbb{A}}\rho_t^c(i,j,a)},~\forall~ 0\leq i, j< N_t.\nonumber
\end{align}
Next, we will show that by tuning the truncation horizon $N_t$, the policy derived from the finite-dimension linear program in Problem \ref{prob3} can be made arbitrarily close to that of the infinite-dimension linear program in Problem 	\ref{prob2}. 
\begin{prop}
	\textbf{($\epsilon$-optimality of the truncated policy):}  Denote the optimal policy derived from the infinite dimension linear program in Problem \ref{prob2} as $\mathbf{\tau}^\star$ and the optimal truncated policy derived from Problem \ref{prob3} as $\mathbf{\tau}^c$.
	Then, we have 
	\begin{align}
	J_e(\mathbf{\tau}^c)	 \leq J_e(\mathbf{\tau}^\star) + \epsilon_g(N_t).\label{gap}
	\end{align}
	In addition, with $\mathbf{\tau}_c$ implemented, the stealthy constraint in (\ref{stealth_refine}) satisfies that
	\begin{align}
	\|\omega_s-\rho_s\|_1\leq \epsilon_s + \epsilon_\omega(N_t).\label{sc_s}
	\end{align}
	Here $\lim_{N_t\rightarrow \infty} \epsilon_g(N_t)=0$ and $\lim_{N_t\rightarrow \infty} \epsilon_\omega(N_t)=0$.\label{eO}
\end{prop}
\begin{proof}
	See Appendix \ref{proof_eO}. 
\end{proof}
Combining Theorem \ref{feasibility} and Proposition \ref{eO}, we can conclude the non-attainability of perfect secrecy as follows.
\begin{coro}
	\textbf{(Non-attainability of ``perfect secrecy")}:
	Given that $\lambda_a<1$ and $\lambda_e\in(1-\frac{1}{|\lambda_{max}(A)|^2},1)$, perfect secrecy cannot be attained for the legitimate estimator since for any reference policy, any detection horizon $N_r$ and any stealthy tolerance $\epsilon_{kl}$, there exists a malicious policy that can satisfy the stealthy constraint and make $J_e$ bounded. 
\end{coro}

Next, we will look at the problem from the legitimate estimator's point of view and explore the design of the optimal reference transmission policy for privacy enhancement. 

\section{Synthesis of the optimal reference policy}
\label{sec5}
In this section, we propose a resilient transmission strategy for the legitimate estimator that can reduce the information leakage to the adversary. Since in practice it is hard to design and implement a policy with parameters of infinite-dimension, we consider a suboptimal policy in the form of
\begin{align}
\mathbb{P}(\nu_{k+1}=1 \mid \eta_k^s=i) = \left\{\begin{array}{cc}\tau_i & i< N_r\\
1 & i\geq N_r\end{array}
\right..\label{policy_t}
\end{align}
It can be verified from Proposition \ref{prop0} that  (\ref{policy_t}) can ensure perfect secrecy for unstable systems before the intrusion, provided that $N_r$ is selected large enough. We will model the interactive decision-making process between the legitimate estimator and the adversary as a Stackelberg game and then transform it into a bilevel program.
\subsection{Formulation of a Stackelberg game}
Consider that the legitimate estimator desires to devise a reference policy that can minimize a linear combination of the AEEC before and after the intrusion, denoted as $J_c$, where
\begin{align}
J_c=\alpha\Big(\beta_1 \hat{J}_{u}- (1-\beta_1) \hat{J}_{e}\Big)+(1-\alpha)\Big(\beta_2 J_{u}- (1-\beta_2) J_{e}\Big),\nonumber
\end{align} 
with  $\alpha,\beta_1,\beta_2\in[0,1]$. 
$J_{u}$ and $J_{e}$ are performance indexes for the legitimate estimator and for the adversary corresponding to the malicious policy, while 
$\hat{J}_{u}$ and $\hat{J}_{e}$ are the performance indexes corresponding to the reference policy.  
 Since the stealthy constraints in (\ref{s1}) and (\ref{s2}) depend on both of the malicious policy and the reference policy, the two problems are coupled. Considering that the legitimate estimator has a priority in determining the reference policy, we can model this interactive decision-making process as a Stackelberg game where the legitimate estimator acts as the leader and the adversary acts as the follower \cite{kar2017trends}, as described in Problem \ref{prob5}. 

\begin{prob}
	\label{prob5}
\begin{align}
\min_{\text{Reference policy}} ~& J_c\nonumber  \\
\text{s.t.}\quad\quad~
&\text{malicious policy} =\arg \min~ J_e\nonumber\\
& \quad\quad\quad\quad\quad\quad\quad\quad\quad\text{s.t.} ~ \quad \text{ stealthy constraint}.\nonumber
\end{align}
\end{prob}
Next, we will detail Problem \ref{prob5} as a bilevel program. 

Let the tensor $\omega_t\in\mathbb{R}^{(N_t+1)\times (N_t+1)\times 2}$ be the occupation measure of the truncated MDP before the intrusion. As discussed in the proof of Theorem \ref{feasibility}, we have \begin{equation}
\begin{aligned}&\omega_t\in\Omega_{N_t}\\
\sum_{0\leq i,j\leq{N}_t,a\in\mathbb{A}}& \omega_t(i,j,a)= 1.
\end{aligned}\label{equation_omega}
\end{equation} and for all $0\leq i, j\leq N_t$, we have
\begin{align}
&\begin{aligned}
 &\frac{\omega_t(i,j,a=1)}{\omega_t(i,j,a=0)+\omega_t(i,j,a=1)} =\tau_i ~~ 0\leq i< N_r   \\
 & \frac{\omega_t(i,j,a=1)}{\omega_t(i,j,a=0)+\omega_t(i,j,a=1)} =1 ~~ N_r\leq i < N_t.
 \end{aligned}\label{uniform},\\
&0\leq \omega_t(i,j,a)\leq 1. \label{in_omega}
\end{align}
We rearrange the tensor $\omega_t$,
into a vector $\omega_b\in\mathbb{R}^{2(N_t+1)^2}$ by defining
\begin{equation}
\omega_b\Big(2(N_t+1)i+2j+a\Big) = \omega_t(i,j,a),\label{vector}
\end{equation} 
such that Eq. (\ref{equation_omega})-(\ref{in_omega}) can be written as
\begin{align}
M_a\omega_b=&m_a\label{eq_1}\\
M_b(\tau)\omega_b=&m_b\label{eq_2}\\
M_c\omega_b\leq& m_c,\label{in_3}
\end{align}
Here $M_a\in\mathbb{R}^{(N_t+1)^2\times 2(N_t+1)^2}$. $M_b(\tau)\in\mathbb{R}^{(N_t+1)^2}$ is dependent on the transmission probability $\tau$ and $M_c\in\mathbb{R}^{4(N_t+1)^2\times 2(N_t+1)^2}$. The vectors $m_a$, $m_b$ and $m_c$ are arranged in proper dimensions. In addition, similar to the derivation of $J_e$ in (\ref{J_e_finite}), we can write $\hat{J}_u$ and $\hat{J}_e$ as linear functions of $\omega_b$:
\begin{equation}
\begin{aligned}
\hat{J}_u=c_u\omega_b \quad \hat{J}_e=c_e\omega_b,
\end{aligned}\label{obj}
\end{equation}
where $c_u,~c_e\in\mathbb{R}^{1\times2(N_t+1)^2}$.

In the same way, we rearrange the tensor $\rho_t$ into a vector $\rho_b$ according to (\ref{vector}).  
Since $\rho_t$ also satisfies (\ref{equation_omega}) and (\ref{in_omega}), we have 
 \begin{align}
 M_a\rho_b=&m_a\label{eqn_1}\\
 M_c\rho_b\leq& m_c.\label{inn_3}
 \end{align}
From Eq. (\ref{policy_trun}),  we have $
 \rho_t(i,N_t,a=0)=0$ and 
 $\rho_t(N_t,j,a=0)=0$
 for all $0\leq i,j\leq N_t$, which are denoted as 
 \begin{equation}
 M_d\rho_t=\mathbf{0},\label{eq}
 \end{equation}
 with $M_d\in\mathbb{R}^{(2N_t+1)\times 2(N_t+1)^2}$. Moreover, by defining $\epsilon=[\epsilon_{0,0},~\epsilon_{0,1},\dots,\epsilon_{N_t,1}]^\top$ and  $\rho_{\epsilon} = [\rho_b^\top,~\epsilon^\top]^\top$, 
 we can write the stealthy constraint in (\ref{s1}) and (\ref{s2}) as
 \begin{equation}
 M_{21}\rho_\epsilon+M_{22}\omega_b\leq m_2,\label{primal_2}
 \end{equation}
 with $M_{21}\in\mathbb{R}^{(N_t+2)\times (2(N_t+1)^2+N_t+1)}$ and $M_{22}\in\mathbb{R}^{(N_t+2)\times 2(N_t+1)^2}$. We may also combine Eq. (\ref{eqn_1})-(\ref{eq}) into
 \begin{equation}
 M_1\rho_{\epsilon} \leq m_1,\label{primal_1}
 \end{equation}
 with $M_1=[M_a^\top,-M_a^\top,M_c^\top,M_d^\top,-M_d]^\top$ and $m_1=[m_a^\top,-m_a^\top,m_c^\top,-m_c^\top,\mathbf{0}_{{\tiny{(4N_t+2)\times 1}}}^\top]^\top$.
Similar to (\ref{obj}), we write the performance indexes of the malicious policy $J_u$ and $J_e$ as
\begin{equation} J_u=c_u\rho_b \quad J_e=c_e\rho_b. \nonumber
\end{equation} Define $\tau=[\tau_0,\dots,\tau_{N_r-1}]$. Problem \ref{prob5} is now cast into a bilevel program with respect to $\tau$, $\omega_b$, and $\rho_{\epsilon}$ ($\rho_b$ is a part of $\rho_\epsilon$).

\begin{prob}
	\label{prob6}
\begin{align}
\min_{\tau,\omega_b,\rho_\epsilon} &~ J_c=\alpha\Big(\beta_1c_u\omega_b-(1-\beta_1)c_e\omega_b\Big)+(1-\alpha)\times\nonumber\\
&\quad\quad~\Big(\beta_2c_u\rho_b-(1-\beta_2)c_e\rho_b\Big).\nonumber\\
\text{s.t.}~&~ 
(\ref{eq_1})-(\ref{in_3}) \nonumber\\
&~[\rho_b^\top ~\epsilon^\top ]^\top= \arg\min_{\rho_\epsilon} ~[~c_e ~\mathbf{0}~]\rho_{\epsilon}\nonumber\\
&\quad\quad\quad\text{s.t.}~~(\ref{primal_2}),(\ref{primal_1}). \nonumber
\end{align}
\end{prob}
\subsection{Properties of the optimal solution}
In this section, we will show that the optimal reference policy is a deterministic policy.

Note that the lower level of Problem \ref{prob6} is linear. Its dual problem should be in the form of
\begin{align}
\max_{y} ~&~ y(\left[\begin{array}{cc} m_1^\top & m_2^\top \end{array}\right]^\top-\left[\begin{array}{cc}\mathbf{0}^\top&M_{22}^\top\end{array}\right]^\top\omega_b) \nonumber\\
\text{s.t.}~~ & ~ y\left[\begin{array}{c}M_1\\M_{21}\end{array}\right]\geq [c_e~\mathbf{0}].\label{dual_1}
\end{align}
According to the duality theorem, the duality gap of this linear program should be zero, i.e. 
\begin{align}y(\left[\begin{array}{cc} m_1^\top & m_2^\top \end{array}\right]^\top-\left[\begin{array}{cc}\mathbf{0}^\top &M_1^\top\end{array}\right]^\top\omega_b) =c_e\rho_b.\label{duality}
\end{align}
Based on this, we can replace the lower level with the inequalities (\ref{primal_1})-(\ref{duality}) and transform Problem \ref{prob6} into a one-level problem as

\begin{prob}
	\label{prob7}
\begin{align}
\min_{\omega_b,\tau,y,\rho_\epsilon} &~ J_c\nonumber\\
\text{s.t.}~~~&~(\ref{eq_1})-(\ref{in_3}), ~(\ref{primal_2})-(\ref{duality}). \nonumber
\end{align}
\end{prob}
Note that Eq. (\ref{duality}) is bilinear with respect to the decision variables $y$ and $\omega_b$. Eq. (\ref{uniform}) is bilinear with respect to $\tau$ and $\omega_b$. 
The objective function and remained constraints are convex. If we define $z_1=[\tau^\top,~y^\top]^\top$ and $z_2=[\omega_b^\top,~\rho_\epsilon^\top]^\top$, we may notice that Problem \ref{prob7} is biconvex with respect to $z_1$ and $z_2$ (Problem \ref{prob7} is convex with $z_1$ if $z_2$ is fixed and vice versa). Moreove, when $z_2$ is fixed, Problem \ref{prob7} becomes a minimization of a concave function. According to \cite{zwart1974global}, the optimal value of $z_1$ are among the extreme points. Since $\tau$ is a component of $z_1$ and both $\tau_i=0$ and $\tau_i=1$ are attainable for all $0\leq i<N_r$. We can conclude that either $\tau_i=0$ or $\tau_i=1$ for all $0\leq i<N_r$. This property is summarized in Proposition \ref{extreme}. 
\begin{prop} \textbf{(Determinacy):}
	The optimal reference transmission policy is a deterministic policy with  $\tau_i=0$ or $\tau_i=1$ for $\forall i\in[0,N_r-1]$. \hfill$\blacksquare$ 
\label{extreme}
\end{prop}

Proposition \ref{extreme} indicates that Problem \ref{prob7} is essentially an integer programming, which has  $2^{N_r}$ candidate solutions. Therefore, the optimal solution can be easily obtained by enumerating all the candidate solutions when $N_r$ is small. In the following section, we will present a branch-and-bound algorithm for the case that $N_r$ is relatively large. 
 
\subsection{A depth-first branch-and-bound algorithm}

Generally, bilevel programming is an NP-hard problem, for a special class of which, the bilevel linear program, penalty function approaches have been invented to derive a global optimum \cite{white1993penalty} or a local optimum \cite{anandalingam1990solution}.
Note that Problem \ref{prob6} is a typical linear bilevel program without the bilinear constraint in Eq.  (\ref{uniform}). In view of this, in this section, we will present a branch-and-bound algorithm to gradually remove the bilinear constraints in (\ref{uniform}), such that Problem \ref{prob7} is in the form of a standard bilevel linear program and can be solved by methods such as \cite{anandalingam1990solution}\cite{white1993penalty}. 

Briefly, we formulate a binary tree with $N_r$ layers and $2^{N_r}$ nodes. In layer $i$, each node denotes a choice of $\tau_i$. As exemplified in Fig. \ref{bab_illus}, the round nodes denote $\tau_i=0$ and the square nodes denote $\tau_i=1$.  Hence, each path in the binary tree gives a feasible solution to Problem \ref{prob7}. With this depth-first branch-and-bound algorithm summarized in Algorithm \ref{alg:branch}, we will search this tree from the left to the right while pruning redundant branches to accelerate the search. 

In Algorithm \ref{alg:branch}, we use 'layer' as an index of the depth of the search. When the search proceeds to $\text{layer}+1$,  $\tau_0,\dots,\tau_{\text{layer}}$ are fixed and $\tau_{\text{layer}+1},\dots,\tau_{N_r-1}$ are to be determined. To decide whether we should go on to visit the subsequent nodes, we need a lower bound on this specific branch (for example, the branch highlighted in blue in Fig. \ref{bab_illus},(a)). Compare this lower bound, denoted as $\underline{J}$, with the candidate optimum $J_{\min}$. If $\underline{J}\geq J_{\min}$, there is no need to search deeper since no solution can outperform the candidate optimum. Otherwise, the subsequent nodes should be visited from the left branch to the right branch in turn. Note that,  if the bilinear constraints in (\ref{uniform}) are removed for $i+1\leq \text{layer} <N_r$, the feasible region is enlarged and the optimal value to the relaxed problem should get smaller, thus, providing a lower bound for the original problem. In particular, we show the relaxed problem in Problem \ref{prob8} which is equivalent to a standard bilevel linear program. 

\begin{prob}
	\label{prob8}
\begin{align}
\underline{J}=\min_{\omega_b,y,\rho_\epsilon} &~ J_c\nonumber\\
\text{s.t.}~~~&~ 
(\ref{eq_1}),(\ref{in_3}),~(\ref{primal_1})-(\ref{duality})\nonumber\\
&~ \tau_i=\frac{\omega(i,j,a=1)}{\omega(i,j,a=0)+\omega(i,j,a=1)} \label{p8}\\
& \quad\quad\quad\quad\quad\quad \quad\quad\quad\quad i=0,\dots,\text{layer},~\forall j\leq N_t.\nonumber
\end{align}
\end{prob}
 The algorithm terminates when the rightmost branch is either traversed or pruned. 
 
 Next, we will show its efficiency with numerical examples in the following section.

\begin{figure}[t]
	\centering
	\includegraphics[width=0.4\textwidth]{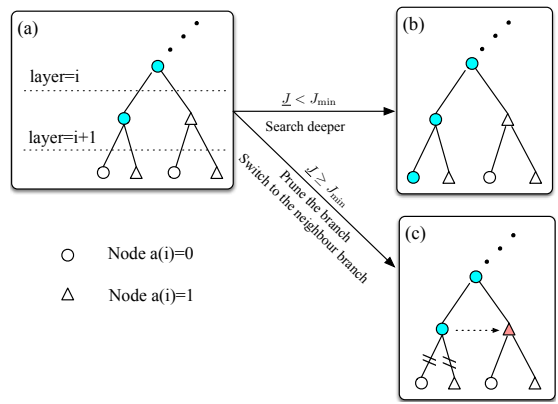}
	\caption{An illustration of the depth-first branch-and-bound algorithm.} \label{bab_illus}
\end{figure}

\begin{algorithm}[h]
	\caption{ A depth-first branch-and-bound method}
	\label{alg:branch}
	\begin{algorithmic}[1]
		\State Initialization: $\text{layer}=0$, $\tau_{\text{layer}}=0$, $\text{end}_\text{id}=0$
		\State Initialize $J_{\min}$ a large number
		\State Set the weighting parameters $\alpha$, $\beta_1$ and $\beta_2$, and the stealthy tolerance $\epsilon_s$
		\While{$\text{end}_\text{id}=0$}
		\While{$\text{layer}<N_r$}
		\State Compute $\underline{J}$ from Problem \ref{prob8}
		\If{$\underline{J}>J_{\min}$} \Comment{\scriptsize pruning}
		\normalsize
		\If{$\min_{i\leq \text{layer}} \tau_i=0$}
		\State $\text{layer}=\arg\max_{\tau_i=0,i\leq\text{layer}} i$
		\State $\tau_{\text{layer}}=1$
 \Comment{\scriptsize switch to the neighbour branch}
		\normalsize
				\State Break
		\Else
		\State Set $\text{end}_{id}=1$ \Comment{\scriptsize the rightmost branch is pruned, exit}
		\normalsize
		\State Break
		\EndIf
		\EndIf
		\State $\text{layer}=\text{layer}+1$ \Comment{\scriptsize search deeper}
		\normalsize
		\State $\tau_{\text{layer}}=0$
		\EndWhile
		\State Compute $\omega_b$ by solving Eq. (\ref{eq_1})-(\ref{in_3}), and $\rho_b$ from the lower level in Problem \ref{prob6}
		\State Compute the objective function $J_c$
		\If{$J_c<J_{\min}$} \Comment{\scriptsize update the candidate optimum}
		\normalsize
		\State Set $J_{\min}=J_c$
		\State Set $\textit{Opt}=\tau$ \Comment{\scriptsize store the candidate optimum}
		\normalsize
		\EndIf
		\If{$\min_{i}\tau(i)=1$}\State Set $\text{end}_{id}=1$ \Comment{\scriptsize the rightmost branch is traversed, exit} 
		\normalsize
		\Else
		\State $\text{layer}=\arg\max_{\tau_i=0} i$ \Comment{\scriptsize switch to the neighbour branch}
		\normalsize
		\State $\tau_{\text{layer}}=1$
		\EndIf
		\EndWhile
		\normalsize
		\State \textbf{Output:} The optimal reference policy \textit{Opt}
	\end{algorithmic}
\end{algorithm}

\section{Numerical example}
\label{sec6}
 Consider the linear system in (\ref{sys}) and (\ref{sys2}) with
\begin{align}
A = & \left[\begin{array}{cc}
1.3 & 1\\
0 & 1
\end{array}\right] \quad C= \left[\begin{array}{cc}
1 & 0
\end{array}\right]\quad Q=\left[\begin{array}{cc}
0.01 & 0 \\
0 & 0.01
\end{array}\right]\nonumber
\end{align}
and $R=0.01$. The steady state error covariance $\bar{P}$ can be computed as
\begin{align}
\bar{P} = \left[\begin{array}{cc}
 0.008 & 0.004\\
0.004 & 0.020
\end{array}\right].\nonumber
\end{align}
Set the packet reception probability as $\lambda=0.6$ at the legitimate estimator and $\lambda_e=0.6$ at the adversary. The reception probability of the acknowledgment is set as $\lambda_a=0.95$. The weighting parameters are taken as $\alpha=0.2$, $\beta_1=0.99$ and $\beta_2=0.5$. Here $\beta_1$ is relatively large since $\hat{J}_e$ is often large as proved in Proposition \ref{prop0}. The optimization horizon $N_r$ is taken as $15$, the truncation horizon $N_t=50$, and the stealthy tolerance $\epsilon_s=0.05$. We apply the proposed branch-and-bound algorithm by solving Problem \ref{prob8} with the penalty function approach in \cite{anandalingam1990solution}. The reference policy is derived after $103$ iterations. The solution presents a threshold structure in the form of
\begin{align}
\nu_{k+1}=\left\{\begin{array}{cc}
0 & \eta_k^s< 6\\
1 & \text{else}
\end{array}\right. .\label{thre}
\end{align}
Corresponding to (\ref{thre}), the optimal malicious transmission policy adopted by the adversary is calculated from Problem \ref{prob3} and plotted in Fig. \ref{malicious_policy}. We may notice that the adversary intends to make the sensor transmit more often when its expected error covariance $P_k^e$ is large. 
\begin{figure}[t]
	\centering
	\includegraphics[width=0.4\textwidth]{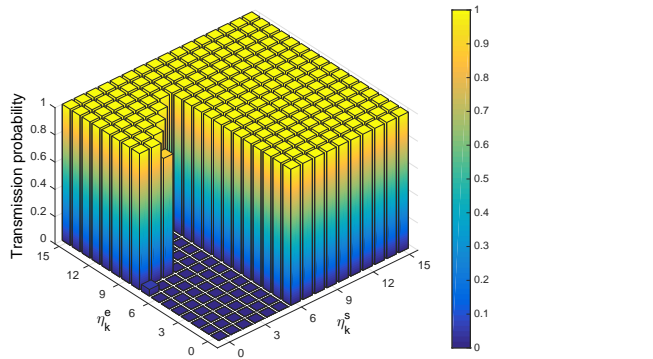}
	\caption{The malicious transmission policy adopted by the adversary. ( $\eta_k^e$ is the holding time at the adversary and $\eta_k^s$ is the holding time at the sensor.)} \label{malicious_policy}
\end{figure}
Next, we would like to investigate the effects of the optimization horizon $N_r$ on the transmission policy and the computation complexity. To this end, we gradually increase $N_r$ from $5$ to $25$. In the first two rows in Table \ref{trun}, we compare the number of iterations required by Algorithm 1 and an exhaustive search. We can see that the number of iterations with the two approaches both increase with $N_r$. However, the proposed branch-and-bound method can efficiently reduce the number of iterations. In row three and row four, we compare the optimal threshold and value of $J_c$ for each $N_r$. We can see that $J_c$ converges with the increase of $N_r$. When $N_r=5$ and $N_r=10$, the proposed algorithm successfully converges to the optimal value confirmed by the exhaustive search.

\begin{table}[h]
	\centering
	\caption{Trade-off between computational complexity and optimality}
	\label{trun}
\scriptsize
\begin{tabular}{ |c||c c c c|  }
	\hline
$\begin{array}{c}	\text{Optimization} \\ \text{horizon}~N_r\end{array}$& $5$ &$10$ &$20$ & $25$\\[5pt]
	\hline
$\begin{array}{c}\textbf{Number of iterations}\\ \text{(Exhaustive search)}
	\end{array}$    & $32$ &$1.0\times 10^3$    &   $1.1\times 10^6$ & $3.4\times 10^7$\\[10pt]
$\begin{array}{c}\textbf{Number of iterations}\\ \text{(Algorithm 1)}
\end{array}$  &8 & 22&    339&     1062  \\[10pt]
$\begin{array}{c}\textbf{Optimal threshold}\\ \text{(Exhaustive search)}
\end{array}$
  & 5   & 6&   N.A. & N.A.\\[10pt]
$\begin{array}{c}\textbf{Optimal threshold}\\ \text{(Algorithm 1)}
\end{array}$ & 5   & 6 &   6 & 6\\[10pt]
\textbf{Optimal value}  & -1.9  &-4.4&   -4.4 &  -4.4\\[5pt]
	\hline
\end{tabular}
\end{table}
\normalsize
\begin{figure*}[htpb]
	\centering
	\includegraphics[width=0.85\textwidth]{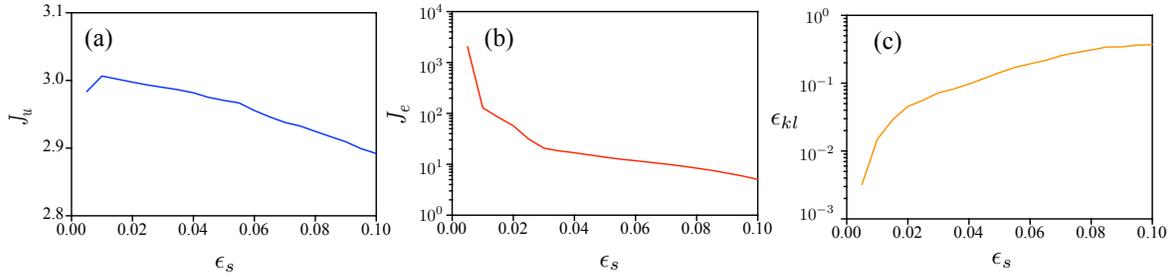}
	\caption{Effects of the stealthy tolerance $\epsilon_s$: (a) the legitimate estimator's performance index $J_u$; (b) the adversary's performance index $J_e$ (c) the Kullback-Leibler divergence $\epsilon_{kl}$.} \label{combo}
\end{figure*}
Next, we would like to investigate the tradeoff between the stealthiness and the estimation performance for the adversary. We fix the truncation horizon as $N_t=50$ and the reference policy as Eq. (\ref{thre}). We vary the stealthy tolerance $\epsilon_s$ from $0$ to $0.1$. The resulted $J_u$ and $J_e$ after the intrusion are plotted in Fig. \ref{combo} (a) and (b). When $\epsilon_s$ increases, both $J_u$ and $J_e$ get smaller. The improvement at the adversary is more significant. Moreover, we calculate the Kullback-Leibler divergence on the distribution of $o_{k+1:k+N_r}$, denoted as $\epsilon_{kl}$, with a Monte Carlo run of length $10^6$. As shown in Fig. \ref{combo}. (c), $\epsilon_{kl}$ monotonically increases with the stealthy tolerance $\epsilon_s$, verifying that restricting $\epsilon_s$ is an efficient way to ensure stealthiness and justifying the stealthy constraint we devised for the adversary. 

\section{Conclusions}

\label{sec7}
In this paper, we study covariance-based transmission policies for remote state estimation, where an active adversary can hack the sensor, reprogram the transmission policy, and overhear the transmissions. From the adversary's perspective, we derive an optimal stealthy malicious policy to re-schedule the transmission such that the estimation performance at the adversary is optimized. This is done by formulating the transmission process as a constrained Markov decision process where an information-theoretic stealthy constraint is incorporated. We show that the feasibility of the constrained MDP depends on the reliability of the ACK channel. If the ACK channel is reliable, there exists a reference policy making the adversary's AEEC unbounded. Otherwise, the constrained MDP is ensured to be feasible with any stealthy tolerance and any reference policy. To make the system resilient, from the legitimate estimator's perspective, we explore the design of the reference policy to maximally reduce information leakage while ensuring its estimation performance. A bilevel program is formulated for this purpose and a depth-first branch-and-bound algorithm is devised to solve the problem. Numerical examples are presented to illustrate the efficacy of the transmission policies as well as 
the efficiency of the proposed optimization algorithm. Future works will include the study of the structural properties of the optimal solutions to further reduce computation cost.

\appendices
\section{Proof of Lemma \ref{prop_ul}}
\label{prlemm2}

An expression of $f^n(\bar{P})$ can be derived from (\ref{def_f}) as  \begin{align}f^{n}(\bar{P})=A^{n}\bar{P}(A^\top)^{n}+\sum_{i=0}^{n-1}A^iQ(A^\top)^i. \label{funf}
\end{align}
For any positive definite matrix $X\in\mathbb{R}^{n_s}$ and any positive integer $i$, we can prove that
\begin{align}
&\text{tr}\left(A^iX(A^\top)^i\right)=\text{tr}\left(X(A^\top)^iA^i\right)\leq n_s \sigma_{\max}\left(X(A^\top)^iA^i\right)\nonumber\\
& \leq n_s\sigma_{\max}(X)\sigma_{\max}^2\left(A^i\right). \nonumber
\end{align}
Moreover, we have
\begin{align}
\text{tr}\left(A^iX(A^\top)^i\right)
&\geq \lambda_{\max}\left(X(A^\top)^iA^i\right) \nonumber\\
&\geq \lambda_{\min}(X)\lambda_{max}\left((A^\top)^iA^i\right). \nonumber
\end{align}
Replace the matrix $X$ with  $\bar{P}$ and $Q$. Eq. (\ref{f_bound}) is derived.
\section{Proof of Proposition \ref{prop_ns}}
\label{ptoP2}
The sufficient condition in Eq. (\ref{sufficient}) can be easily proved from Theorem 2 in \cite{schenato2008optimal}. The necessary condition in (\ref{lim}) can be derived from Eq. (\ref{pi_e}). Since
	\begin{align}
\sum_{j=0}^\infty \pi_j \text{tr}(f^j(\bar{P}))=&
\frac{1}{1+\sum_{l=0}^\infty \prod_{i=0}^l(1-\tau_i\lambda)}\text{tr}(f^0(\bar{P}))+ \nonumber\\
& \frac{\sum_{l=1}^\infty\prod_{i=0}^{l-1}(1-\tau_i\lambda)\text{tr}(f^l(\bar{P}))}{1+\sum_{l=0}^\infty \prod_{i=0}^l(1-\tau_i\lambda)},\label{finite}
\end{align}
 if $|\lambda_{max}(A)|\geq 1$, $\text{tr}(f^l(\bar{P}))$ increase linearlly or exponentially to infinity according to Lemma \ref{prop_ul}. Moreover, since $J_u$ is bounded, the term $\sum_{l=1}^\infty\prod_{i=0}^{l-1}(1-\tau_i\lambda)\text{tr}(f^l(\bar{P}))$ must be bounded. Thus,  $\lim_{l\rightarrow\infty} \prod_{i=0}^{l-1}(1-\tau_i\lambda)= 0$. 
\section{Proof of Proposition \ref{prop0}} \label{ptprop3}

	(i) Consider a marginally stable system where $|\lambda_{\max}(A)|=1$. Lemma \ref{prop_ul} suggests that $\text{tr}(f^i(\bar{P}))$ linearly increases with $i$. Hence, there exist two scalars $0<c_l<c_u$ and $c_0>0$ such that
	\begin{align}
	c_0+c_li\leq\text{tr}(f^i(\bar{P})) \leq c_0+c_ui,~~\forall i\in\mathbb{N}^+.\label{linear_bound}
	\end{align}
From Eq. (\ref{lim}), we know that to ensure boundedness of $J_u$, there must exist a positive integer $N$ such that the sensor transmits with a probability greater than zero at least once between $k=mN$ and $k=(m+1)N$ for all $m\in\mathbb{N}$. 
To address the worst scenario, we detail the transmission policy as for $k\in[mN,(m+1)N]$, the sensor transmits once with probability $\epsilon_t>0$. With this strategy, the probability that all transmissions are failed at the adversary is $(1-\lambda_e\epsilon_t)^{\lfloor k/N\rfloor}$. Then, from Eq. (\ref{linear_bound}), one can derive that
\begin{align}
\mathbb{P}(\text{tr}\mathbb{E}(P_k^e) \geq c_0+c_uk)\leq (1-\lambda_e\epsilon_t)^{\lfloor k/N\rfloor}.\label{tail}
\end{align}
Using the tail sum formula, we have
	\begin{align}
	\mathbb{E}\left(\frac{\text{tr}\mathbb{E}(P_k^e)-c_0}{c_u}\right) \leq \sum_{k=0}^\infty (1-\lambda_e\epsilon_t)^{\lfloor k/N\rfloor}.\nonumber
	\end{align}
	Since $0<1-\lambda_e\epsilon_t<1$, $ \mathbb{E}\left(\frac{\text{tr}\mathbb{E}(P_k^e)-c_0}{c_u}\right)$ is bounded. Thus, $\text{tr}\mathbb{E}(P_k^e)$ is bounded. Denote this bound as $\bar{b}_e$. It can be easily proved that 
$J_e=\lim_{T\rightarrow\infty}\frac{1}{T} \sum_{k=1}^T\text{tr}(\mathbb{E}(P_k^e)) <\bar{b}_e<\infty$.  

To show that $J_e$ can be arbitrarily large, we assume that the sensor adopts a threshold-type policy as  
	\begin{align}
\mathbb{P}(\nu_{k+1}=1\mid \eta_k^s)=\left\{\begin{array}{cc}
1 & \eta_k^s > \bar{t}\\
0 & \text{else}
\end{array}\right. .\label{thres}
\end{align}
Using the expressions in Eq. (\ref{pi_e}), the stationary probability of the event $\eta_k^s = i,~i\in[0,\bar{t}]$ is $\frac{1}{2+\bar{t}+\frac{1-\lambda}{\lambda}}.$ 
Consider an event $\omega$ where all transmissions made by the sensor are successfully received by the adversary. Since this is an ideal scenario which can only be achieved when $\lambda_e=1$, it provides the minimum $J_e$. Therefore, we have a lower bound of $J_e$ as
\begin{align}
J_e\geq \frac{1}{2+\bar{t}+\frac{1-\lambda}{\lambda}}
\sum_{i=0}^{\bar{t}} c_0+c_li
=\frac{(c_0+c_l\bar{t}/2)(1+\bar{t})}{2+\bar{t}+\frac{1-\lambda}{\lambda}}.\nonumber\end{align}
Hence, by tuning the threshold $\bar{t}$, $J_e$ can be made arbitrarily large.
	
	(ii) If the system is unstable, according to Theorem III.4 in \cite{leong2018transmission}, a threshold type policy in the form of Eq. (\ref{thres}) can make $J_e$ infinite if the threshold $\bar{t}$ satisfies that 
 \begin{align}
\lambda_e < 1-\frac{1}{\lambda\lambda_a | \lambda_{\max}(A)|^{2(\bar{t}+1)}}.\label{condition}
\end{align}
The existence of the threshold is ensured if $\lambda_e<1$.

\section{Proof to Theorem \ref{safe}}

\label{pr_th6}
If $\lambda_a=1$, we have $\eta_k=\eta_k^s$ from Eq. (\ref{eta}), (\ref{gamma}) and (\ref{eta_s}). Given that $N_r>\bar{t}$ and the reference policy is set as (\ref{candidate}), the transmission probability $\tilde{\tau}_{ij}$ must be $0$ for all $j$ when $i\leq \bar{t}$. Otherwise, there will be some observations in the form of $o_k=\{1,\underbrace{0,~0,\dots,0}_{\text{number of zero}<\bar{t}},1\}$, which will never show up when the policy (\ref{candidate}) is applied. 
This will make the Kullback-Leibler divergence in (\ref{stealthy_kl}) not well-defined and violate the stealthy constraint in (\ref{stealthy_kl}). 
Next, we will show that in the case that $\tilde{\tau}_{ij}=0$ for all $j$ and $i\leq \bar{t}$, $J_e$ is unbounded. Since the adversary intends to minimize $J_e$, from Theorem III.3 in \cite{leong2018transmission}, there exists a threshold $N_l$ such that $\tilde{\tau}_{ij}=1$ for $i\geq \bar{t}$ and $j\geq N_l$. Then, similar to the proof to Theorem III.6 in \cite{leong2018transmission}, we consider an event $\omega_e$: after $j$ reaches $N_l$, every transmission is successfully received by the legitimate estimator and unsuccessfully received by the adversary. Then, we have
\begin{align}
J_e>& \mathbb{P}(\omega_e)\frac{1}{K}\sum_{k=1}^\infty \text{tr} \mathbb{E}(P_k^e\mid \omega_e)\nonumber\\
>& \mathbb{P}(j=N_l)\lim_{K\rightarrow\infty}\frac{1}{K}\text{tr}\big(A^K\bar{P}(A^K)^\top\big)\big(\lambda(1-\lambda_e)\big)^{K/(\bar{t}+1)}.\nonumber
\end{align}
From Eq. (\ref{condition}), we may notice that the term $\text{tr}\big(A^K\bar{P}(A^K)^\top\big)\big(\lambda(1-\lambda_e)\big)^{K/(\bar{t}+1)}$ exponentially increases with $K$. Therefore, $J_e$ is unbounded.

\section{Proof of Proposition \ref{al_stealthy}}
\label{proof_pi}

Define a set $\mathbb{O}=\Big\{(\nu_{k+1}\gamma_{k+1},\dots,\nu_{k+N_r}\gamma_{k+N_r})~|~ \nu_{k+i}\gamma_{k+i}\in\{0,1\},i\in[1,N_r]\Big\}.$
First, we will prove that if $\lambda_a<1$, the probability of any observation $o_{k+1:k+N_r}$ in the set $\mathbb{O}$ is greater than $0$, i.e. $\mathbb{P}(o_{k+1:k+N_r})>0$ for any $o_{k+1:k+N_r}\in\mathbb{O}$.
To this end, we will derive the stationary distribution of the pair $(\eta_k,\eta_k^s)$ by modeling its evolution as a Markov chain. Take $s_k=(\eta_k,\eta_k^s)$. From
Eq. (\ref{trans1}), (\ref{eta}) (\ref{gamma}), (\ref{policy}), and (\ref{eta_s}), the transition probability can be derived as
\begin{align}
&\mathbb{P}(s_{k+1}\mid s_k)=\left\{\begin{array}{ll}
\tau_j\lambda\lambda_a & s_{k+1}=(0,0)\\
\tau_j\lambda(1-\lambda_a) & s_{k+1}=(0,j+1) \\
1-\tau_j\lambda & s_{k+1}=(i+1,j+1)
\end{array}\right. ,\label{trans5}
\end{align}
with $s_k=(i,j)$ and $i,j\in\mathbb{N}$. Denote the stationary probability of the state $s_k=(i,j)$ as $\pi^c_{i,j}$. From Eq. (\ref{trans5}), we have
\begin{equation}
\begin{aligned}
\pi_{0,0}^c = & \sum_{i,j\in\mathbb{N}} \pi_{i,j}^c\tau_j\lambda\lambda_a\\
\pi_{0,j+1}^c = & \sum_{i\in\mathbb{N}} \pi_{i,j}^c\tau_j\lambda(1-\lambda_a) \\
\pi_{i+1,j+1}^c = &  \pi_{i,j}^c(1-\tau_j\lambda)
\end{aligned},\label{stationary}
\end{equation}
which  indicates that $\pi_{i,j}^c>0$ for all $j\geq i\geq0$. Then, we have $\mathbb{P}(\eta_k=i)=\sum_{j\in\mathbb{N}}\pi_{i,j}^c>0$ for any $i\in\mathbb{N}$. Since $\mathbb{P}(\eta_k=i,\nu_{k+1}\gamma_{k+1}=0)=\sum_{j=i}^\infty\pi_{i,j}^c(1-\tau_{j}\lambda)$ and  $\mathbb{P}(\eta_k=i,\nu_{k+1}\gamma_{k+1}=1)=\sum_{j=i}^\infty\pi_{i,j}^c\tau_{j}\lambda$, we have \begin{equation}\mathbb{P}(\eta_k=i,\nu_{k+1}\gamma_{k+1}=0)>0,\quad \mathbb{P}(\eta_k=i,\nu_{k+1}\gamma_{k+1}=1)>0. \label{nonzero}
\end{equation}
 Note from Eq. (\ref{eta}) that
\begin{align}
&\mathbb{P}(o_{k+1:k+N_r}=(i_1,\dots,i_{N_r})) \nonumber\\
= &\sum_{l=0}^\infty \mathbb{P}(\eta_k=l)\mathbb{P}(\nu_{k+1}\gamma_{k+1}=i_1|\eta_{k}=l)\dots\times\nonumber\\
&\mathbb{P}\Big(\nu_{k+N_r}\gamma_{k+N_r}=i_{N_r}|\eta_{k}=l,\nu_{k+j}\gamma_{k+j}=i_{j},j\in[1,N_r)\Big),\label{ob}
\end{align}
with $i_1,\dots,i_{n_r}\in\{0,1\}$. Combining 
Eq. (\ref{nonzero}) and (\ref{ob}), we have  $\mathbb{P}(o_{k+1:k+N_r}=(i_1,\dots,i_{N_r}))>0$ for any $(i_1,\dots,i_{N_r})$.  This ensures that the Kullback-Leibler divergence is always well-defined. We can further derive from Eq. (\ref{stationary}) that 
\begin{equation}
\begin{aligned}
&\mathbb{P}(\eta_k=i,\nu_{k+1}\gamma_{k+1}=0)= \\
& \sum_{j\in\mathbb{N}} \omega_s(j,1)\lambda\left[\lambda_a\prod_{l=0}^{i}(1-\tau_l\lambda)+(1-\lambda_a)\prod_{l=1}^{i+1}(1-\tau_{j+l}\lambda)\right]\\
&\mathbb{P}(\eta_k=i,\nu_{k+1}\gamma_{k+1}=1)= \sum_{j\in\mathbb{N}} \omega_s(j,1)\lambda\times\\
&\left[\lambda_a\prod_{l=0}^{i-1}(1-\tau_l\lambda)\tau_i\lambda+(1-\lambda_a)\prod_{l=1}^{i}(1-\tau_{j+l}\lambda)\tau_{j+i+1}\lambda\right]
\end{aligned} \label{pi}
\end{equation}
where $\omega_s(j,1)$ is the stationary probability of $(\eta_k^s=j,\nu_{k+1}=1)$ and $\tau_l =\frac{\omega_s(l,1)}{\omega_s(l,0) + \omega_s(l,1)}$ for $\forall l\in\mathbb{N}$. 
With Eq. (\ref{pi}) and (\ref{ob}), we may see that the mapping between $\mathbb{P}(\eta_k,\nu_{k+1}\gamma_{k+1})$ and $\omega_s$ is continuously differentiable. Therefore, the mapping is locally Lipschitz continuous, i.e. there exists a constant $\mathcal{L}$ such that $\|\mathbb{P}(o_k=(i_1,\dots,i_{N_t}))-\mathbb{P}(o_k'=(i_1,\dots,i_{N_t})))\|_1\leq\mathcal{L}\|\omega_s-\rho_s\|_1$ where $o_k'$ denotes the observations after the intrusion. Since $\mathbb{P}(o_k=(i_1,\dots,i_{N_t}))>0$, the KL divergence can be made arbitrarily small by squeezing $\|\omega_s-\rho_s\|_1$, which completes the proof.

\section{Proof of Proposition \ref{prop_8}}
\label{proof_prop8}
According to the definition of $\rho$ and $\rho_t$, we have
\begin{equation}
\begin{aligned}
& \rho_t(i,j,a) = \rho(i,j,a)\quad 0\leq i,j<N_t \\
& \rho_t(i,N_t,a=1) = \sum_{j=N_t}^\infty \rho(i,j,a=1),\quad 0\leq i<N_t\\
& \rho_t(N_t,j,a=1) = \sum_{i=N_t}^\infty \rho(i,j,a=1),\quad 0\leq j <N_t \\
& \rho_t(N_t,N_t,a=1) = \sum_{i=N_t}^\infty\sum_{j=N_t}^\infty \rho(i,j,a=1),
\end{aligned} \label{rho_t}
\end{equation}
which further gives that
\begin{align}
&\sum_{i=0}^\infty \sum_{a\in\mathbb{A}}\rho(i,j+1,a) =  (1-\lambda_e) \sum_{i=0}^\infty \sum_{a\in\mathbb{A}}\rho(i,j,a),~j\geq N_t \nonumber\\
& \sum_{i=0}^{N_t}\rho_t(i,N_t,a=1) = \sum_{i=0}^\infty \sum_{j=N_t}^\infty \sum_{a\in\mathbb{A}} \rho(i,j,a).\nonumber
\end{align}
When $j\geq N_t$, we have $\sum_{i=0}^\infty \sum_{a\in\mathbb{A}}\rho(i,j,a)=\lambda_e(1-\lambda_e)^{j-N_t} \sum_{i=0}^{N_t}\rho_t(i,N_t,a=1)$.  The averaged expected error covariance can be expressed as
\begin{align}
&\lim_{T\rightarrow\infty}\frac{1}{T} \sum_{k=1}^T\text{tr}(\mathbb{E}(P_k^e)) = \sum_{j=0}^\infty\sum_{i=0}^\infty\sum_{a\in\mathbb{A}}\rho(i,j,a)\text{tr}(f^j(\bar{P}))\nonumber\\
= & \sum_{j\in[0,N_t-1]}\sum_{i\in[0,N_t]}\sum_{a\in\mathbb{A}}\rho_t(i,j,a)\text{tr}(f^j(\bar{P}))+\nonumber\\
&\sum_{i\in[0,N_t]}\rho_t(i,N_t,a=1)\sum_{j=0}^\infty\lambda_e(1-\lambda_e)^{j-N_t}\text{tr}f^j(\bar{P}).\nonumber
\end{align}
Since $\lambda_e>1-\frac{1}{|\lambda_{\max}(A)|^2}$, the term $\sum_{j=0}^\infty\lambda_e(1-\lambda_e)^{j-N_t}\text{tr}f^j(\bar{P})$ is bounded for any $N_t$ from Lemma \ref{prop_ul}.
	\section{Proof of Theorem \ref{feasibility}}
\label{prtoThm} 
We devise a candidate policy in the form of 
\begin{align}
\tau_{ij}^t = \left\{\begin{array}{ll}
{\tau}_{i} & 0 \leq i,j< N_t \\
1 & i\geq N_t ~\text{or}~ j\geq N_t
\end{array}\right., \label{candidate_trun}
\end{align}
where $\tilde{\tau}_{ij}$ is set the same as the transmission probability $\tau_i$ in the reference policy when $0\leq i,j<N_t$. We will show that with this policy applied, there exists a $N_t$ making $J_e$ bounded and the stealthy constraint satisfied. Since Proposition \ref{prop_8} has shown the boundeness, we will focus on the proof of constraint satisfaction. 

Define a tensor $\omega_t\in\mathbb{R}^{(N_t+1)\times (N_t+1)\times 2}$ as the occupation measures of the truncated MDP when the reference policy in (\ref{policy}) is applied. Then, we have $\omega_t\in\Omega_{N_t}$ and \begin{equation}\tau_i=\frac{\omega_t(i,j,a=1)}{\sum_{a\in\mathbb{A}}\omega_t(i,j,a)}\end{equation} for all $0\leq i,j< N_t$. Similar to (\ref{rho_t}), we have
\begin{align}
\omega_t(i,N_t,a=0)=\sum_{j=N_t}^\infty \omega(i,j,a=0)\\
\omega_t(N_t,j,a=0)=\sum_{i=N_t}^\infty \omega(i,j,a=0)
\end{align}
for $0\leq i,j\leq N_t$. We can rearrange the tensor $\omega_t$ into a vector $\omega_b$ such that 
\begin{equation}
\omega_b\Big(2(N_t+1)i+2j+a\Big) = \omega_t(i,j,a),\label{omega_b}
\end{equation} 
and rewrite the above equalities in $\omega_t$ compactly as 
\begin{equation}
M_c \omega_b = N_\omega, \label{equality}
\end{equation}
with $M_c\in\mathbb{R}^{2(N_t+1)^2\times 2(N_t+1)^2}$ being of full rank and thus invertible, and $N_\omega \in\mathbb{R}^{2(N_t+1)^2}$. Similarly, we can prove that $\rho_t\in\Omega_{N_t}$.To be consistent with Eq. (\ref{candidate_trun}), $\rho_t$ must satisfy that \begin{equation}\tau_i=\frac{\rho_t(i,j,a=1)}{\sum_{a\in\mathbb{A}}\rho_t(i,j,a)}\end{equation} for all $i,j< N_t$ and
\begin{align}
\rho_t(i,N_t,a=0)=0\label{sign1}\\
\rho_t(N_t,j,a=0)=0\label{sign2}
\end{align}
for $i,j\leq N_t$. By defining $\rho_b$ as 
\begin{equation}
\rho_b\Big(2(N_t+1)i+2j+a\Big) = \rho_t(i,j,a),\nonumber
\end{equation} 
we may notice that the equality constraints on $\rho_b$ can be written as
\begin{equation}
M_c\rho_b=N_\rho 
\end{equation}
where $N_\rho$ can be adapted from $N_\omega$ by making the items corresponding to (\ref{sign1}) and (\ref{sign2}) from $\sum_{j=N_t}^\infty \omega(i,j,a=0)$ and $\sum_{i=N_t}^\infty \omega(i,j,a=0)$ to $0$. If we can prove that $\sum_{j=N_t}^\infty \omega(i,j,a=0)$ and $\sum_{i=N_t}^\infty \omega(i,j,a=0)$ converge to $0$ with $N_t$, then we have $\lim_{N_t\rightarrow\infty} \|N_\omega-N_\rho\|_1=0$. Since $M_c$ is invertible, we can prove that for any $\epsilon_s$, there exists a $N_t$ such that $
\|\rho_b-\omega_b\|_1\leq\epsilon_s.$

Since the convergence of $\sum_{i=N_t}^\infty \omega(i,j,a=0)$ can be easily obtained from (\ref{pi_e}), we focus on the proof of $\lim_{N_t\rightarrow\infty}\sum_{j=N_t}^\infty \omega(N_t,j,a=0)=0$. Define the stationary distribution of the state $(\eta_k^s=i,\eta_k^e=j)$ with $i,j\in\mathbb{N}$ as $\pi_{i,j}$ before the intrusion. Then, $\omega(i,j,a=0)\leq \pi_{i,j}$ for any $i,j\in\mathbb{N}$.
From Eq. (\ref{policy}) and Eq. (\ref{trans6}), we have
\begin{equation}
\begin{aligned}
\pi_{0,0} = & \lambda\lambda_e \sum_{j=0}^{\infty}\sum_{i=0}^\infty \tau_{i} \pi_{i,j} \\
\pi_{0,j+1} = & \lambda(1-\lambda_e) \sum_{i=0}^\infty \tau_{i}\pi_{i,j} \\
\pi_{i+1,0} = & (1-\lambda)\lambda_e \sum_{j=0}^\infty \tau_{i}\pi_{i,j} \\
\pi_{i+1,j+1} = & [(1-\lambda)(1-\lambda_e)\tau_{i}+1-\tau_{i}]\pi_{i,j}.
\end{aligned} \label{stationary1}
\end{equation}
Define $\pi_{j}^e=\sum_{i=0}^\infty \pi_{i,j}$, we have
\begin{align}
\pi_{j+1}^e = & \pi_{0,j+1}+\sum_{i=0}^\infty \pi_{i+1,j+1} \nonumber\\
= & \lambda\lambda_a(1-\lambda_e)\sum_{i=0}^\infty \tau_i\pi_{i,j} +\sum_{i=0}^\infty [(1-\lambda_a\lambda)(1-\lambda_e)\tau_i\nonumber\\
&+1-\tau_i]\pi_{i,j}\nonumber\\
= & \underbrace{\sum_{i=0}^\infty \pi_{i,j}}_{\pi_j^e}-\lambda_e\underbrace{\sum_{i=0}^\infty \tau_i\pi_{i,j}}_{\frac{\pi_{0,j+1}}{\lambda_a\lambda(1-\lambda_e)}}
=  \pi_{j}^e -\frac{\lambda_e\pi_{0,j+1}}{\lambda\lambda_a(1-\lambda_e)}.
\end{align}
Since $\pi_{0,j+1}\geq0$, we have $\pi_j^e\geq \pi_{j+1}^e$. Considering that $\pi_j^e\geq0$ and $\sum_{j=0}^\infty \pi_j^e=1$, $\pi_j^e$ is a Cauchy sequence and, thus, $\lim_{N_t\rightarrow\infty}\sum_{j=N_t}^\infty \pi_j^e=0$. Moreover, we have $\lim_{N_t\rightarrow\infty}\sum_{j=N_t}^\infty \omega(N_t,j,a=0)\leq\lim_{N_t\rightarrow\infty}\sum_{j=N_t}^\infty\pi_j^e=0 $. The proof is completed.
\end{proof}
	
\section{Proof of Proposition \ref{eO}}
\label{proof_eO}
According to the definition of $\rho_t$, Eq. (\ref{s1}) and (\ref{s2}) are equivalent to $\sum_{i=0}^{N_t-1}\|\omega_s(i,a)-\rho_s(i,a)\|_1+\|\sum_{N_t}^\infty \omega_s(i,a)-\sum_{N_t}^\infty \rho_s(i,a)\|_1\leq \epsilon_s$. We can prove from (\ref{pi_e}) that $\lim_{i=N_t\rightarrow\infty}\sum_{a\in\mathbb{A}}\sum_{i=N_t}^\infty\omega_s(i,a)=0$ and $\lim_{N_t\rightarrow\infty}\sum_{N_t}^\infty\sum_{a\in\mathbb{A}}\rho_s(i,a)=0$. Since  \begin{equation}\begin{aligned}&\|\omega_s-\rho_s\|_1\leq\\ &\underbrace{\sum_{i=0}^{N_t-1}\|\omega_s(i,a)-\rho_s(i,a)\|_1}_{\leq \epsilon_s} +\underbrace{\sum_{a\in\mathbb{A}}\sum_{N_t}^\infty\omega_s(i,a)+  \rho_s(i,a)}_{\text{converge to}~ 0 ~\text{with} ~N_t},\end{aligned} \label{p9}
\end{equation} Eq. (\ref{sc_s}) is proved.  

To prove Eq. (\ref{gap}), we may devise a policy $\tau^f$ as
\begin{align}
\tau^f_{ij}=\left\{\begin{array}{cc}
\tau^\star_{ij} & 0\leq i,j< N_t\\
1 & i\geq N_t ~\text{or}~ j\geq N_t
\end{array}\right..\nonumber
\end{align}
From Proposition \ref{prop_8} and Theorem \ref{feasibility}, both $\tau^f$ and $\tau^\star$ give bounded $J_e$. Therefore, the discrepancy between $\tau^f$ and $\tau^\star$ in $J_e$ diminishes with $N_t$, i.e.
\begin{equation}|J_e(\tau^f)-J_e(\tau^\star)|\leq 
\epsilon_a(N_t)\label{eq1}\end{equation} where $\lim_{N_t\rightarrow\infty}\epsilon_a(N_t)=0$. Denote the value of $\sum_{i=0}^{N_t-1}\|\omega_s(i,a)-\rho_s(i,a)\|_1+\|\sum_{N_t}^\infty \omega_s(i,a)-\sum_{N_t}^\infty \rho_s(i,a)\|_1$ as $\epsilon_u(N_t)$. From (\ref{p9}), we have $\lim_{N_t\rightarrow\infty}|\epsilon_u(N_t)-\epsilon_s|=0$. Thus, the policy $\tau^f$ can be viewed as a feasible solution to Problem	\ref{prob2} with the stealthy tolerance $\epsilon_s$ perturbed with $|\epsilon_u(N_t)-\epsilon_s|$. Based on the local sensitivity analysis for a convex problem (Section 5.6 in \cite{boyd2004convex}), we have \begin{equation}J_e(\tau^c)\leq J_e(\tau^f)+\epsilon_b(N_t)\label{eq2}\end{equation}
with $\lim_{N_t\rightarrow\infty}\epsilon_b(N_t)=0$. Define $\epsilon_g(N_t)=\epsilon_a(N_t)+\epsilon_b(N_t)$. Eq. (\ref{eq1}) and (\ref{eq2}) together give Eq. (\ref{gap}). The proof is complete.

\bibliographystyle{ieeetran}
\bibliography{stealthy.bib}

\end{document}